\documentclass[twocolumn,showpacs,preprintnumbers,amsmath,amssymb]{revtex4}
\usepackage{algorithmic}
\usepackage{algorithm}
\usepackage{ amsmath}
\usepackage{graphicx}
\usepackage{dcolumn}
\usepackage{bm}
\newcommand{\figref}[1]{Figure~\ref{#1}}
\newcommand{\remove}[1]{} 
\newtheorem{theorem}{Theorem}

\begin{document}

\preprint{APS/123-QED}

\title{A Parameterized Centrality Metric for Network Analysis}

\author{Rumi Ghosh}%
 \email{rumig@usc.edu}
\author{Kristina Lerman}
 \email{lerman@isi.edu}
\affiliation{%
USC Information Sciences Institute\\
4676 Admiralty Way, Marina del Rey, CA 90292
}%

\date{\today}

\begin{abstract}
A variety of metrics have been proposed to measure the relative importance of nodes in a network.
One of these, $\alpha$-centrality~\cite{Bonacich:2001}, measures the number of attenuated paths that exist between nodes.
We introduce a normalized version of this metric and use it to study network structure, specifically, to rank nodes and find community structure of the network. Specifically, we extend the modularity-maximization method~\cite{GirvanNewman04} for community detection to use this metric as the measure of node connectivity. Normalized $\alpha$-centrality is a powerful tool for network analysis, since it contains a tunable parameter that sets the length scale of interactions. By studying how rankings and discovered communities change when this parameter is varied allows us to identify locally and globally important nodes and structures.
We apply the proposed method to several benchmark networks and show that it leads to better insight into network structure than alternative methods.
\end{abstract}

\maketitle

\section{Introduction}\label{sec:intro}
Centrality measures the degree to which network structure contributes to the importance, or status, of a node in a network. Over the years many different centrality metrics have been defined.
One of the more popular metrics, betweenness centrality~\cite{Freeman}, measures the fraction of all shortest paths in a network that pass through a given node. Other centrality metrics include those based on random walks~\cite{Stephenson89,PageRank,Noh02,Newman05} and path-based metrics. The simplest path-based metric, degree centrality, measures the number of edges that connect a node to others in a network. According to this measure, the most important nodes are those that have the most connections. 
However, a node's centrality depends not only on how many others it is connected to but also on the centralities of those nodes~\cite{Katz,Bonacich:2001}. This measure is captured by the total number of paths linking a node to other nodes in a network. 
One such metric, $\alpha$-centrality~\cite{Bonacich87,Bonacich:2001}, measures the total number of paths from a node, exponentially attenuated by their length.
The attenuation parameter sets the length scale of interactions. Unlike other centrality metrics, which do not distinguish between local and global structure, a parameterized centrality metric can differentiate between locally connected nodes, i.e., nodes that are linked to other nodes which are themselves interconnected, and globally connected nodes that link and mediate communication between poorly connected groups of nodes. Studies of human~\cite{Granovetter,Simmel,Csermely08} and animal~\cite{Lusseau04} populations suggest that such `bridges' or `brokers' play a crucial role in the information flow and cohesiveness of the entire group.

One difficulty in applying $\alpha$-centrality in network analysis is that its key parameter is bounded by the spectrum of the corresponding adjacency matrix of the network. As a result, the metric diverges for larger values of this parameter. We address this problem by defining \emph{normalized $\alpha$-centrality}. We show that the new metric avoids the problem of bounded parameters while retaining the desirable characteristics of $\alpha$-centrality, namely its ability to differentiate between local and global structures.

In addition to ranking nodes, parameterized centrality can be used to identify communities within a network~\cite{Fortunato10}. In this paper, we generalize modularity maximization-based approach~\cite{Newman104,Newman206}  to use normalized $\alpha$-centrality. Rather than find regions of the network that have greater than expected number of edges connecting nodes~\cite{GirvanNewman04}, our approach looks for regions that have greater than expected number of weighted paths connecting nodes. One advantage of this method is that the attenuation parameter can be varied to identify  local vs.\  global communities.

Normalized $\alpha$-centrality is a powerful tool for network analysis. By  differentiating between locally and globally connected nodes, it provides a simple alternative to previous attempts to quantify fine-grained structure of complex networks, such as the motif-based~\cite{Milo02,Milo04} and role-based~\cite{Guimera05,Guimera07} descriptions. The former measures the relative abundance of subgraphs of a certain type, while latter classifies nodes according to their connectivity within and outside of their community. Applying either of these descriptions to real networks is computationally expensive: role-based analysis, for example,  requires the network to be decomposed into distinct communities first. Normalized $\alpha$-centrality, on the other hand, measures node connectivity at different length scales, allowing us to resolve network structure in a computationally efficient manner.

We use normalized $\alpha$-centrality to study the structure of several benchmark networks, as well as a real-world online social network. We show that this parameterized centrality metric can identify locally and globally important nodes and communities, leading to a more nuanced understanding of network structure.

\section{Centrality and Network Structure}
\label{sec:centrality}
Bonacich~\cite{Bonacich87,Bonacich:2001} defined \emph{$\alpha$-centrality} $C_{i,j}(\alpha,\beta, n)$ as the total number of attenuated paths between nodes $i$ and $j$, with $\beta$ and $\alpha$ giving the attenuation factors\footnote{For some types of networks, e.g., commodity exchange networks, Bonacich allows $\alpha<0$. In communication and information networks we are considering, $\alpha>0$.} along direct edges (from $i$) and indirect edges (from intermediate nodes) in the path from $i$ to $j$, respectively, and $n$ is the length of the longest path.
Given the adjacency matrix of the network $A$, $\alpha$-centrality matrix is defined as follows:
\begin{equation}
\label{eq:inf}
C(\alpha,\beta,n) = \beta A +\beta \alpha_1 A^2 + \cdots + \beta \prod_{k=1}^n\alpha_k   A^{n+1}
\end{equation}
\noindent
The first term gives the number of paths of length one (edges) from $i$ to $j$, the second gives the number of paths of length two, etc. Although $\alpha_k$ along different edges in a path could in principle be different, for simplicity, we take them all to be equal: $\alpha_k=\alpha, \ \forall k$. In this case, the series converges to $C(\alpha,\beta, n\to \infty)=\beta A {( I-\alpha A)}^{-1}$, which holds while  $\alpha < 1/\lambda_1$, where $\lambda_1$ is the largest characteristic root of  $A$~\cite{Ferrar}. The computation of $\lambda_1$ is difficult, especially for large networks, which include most complex real-world networks.

To get around this difficulty, we define \emph{normalized $\alpha$-centrality} matrix as:
\begin{equation}
\label{eq:inf1}
NC(\alpha,\beta,n \to \infty) =\frac{C(\alpha,\beta,n \to \infty)}{\sum_{ij} C_{ij}(\alpha,\beta, n \to \infty)}
\end{equation}
\noindent
As we show in  the appendix, in contrast to $\alpha$-centrality, normalized $\alpha$-centrality is not bounded by $\lambda_1$. Also, we prove that,  assuming $|\lambda_1| $ is strictly greater than any other eigenvalue, $\lim_{\alpha  \to 1/{|\lambda_{1}|}}NC(\alpha,\beta, n \to \infty)$ exists; and as $\alpha$ is increased, $NC(\alpha,\beta, n \to \infty)$ converges to this value and is finite for $\alpha \le 1$.

Just like the original $\alpha$-centrality, normalized $\alpha$-centrality contains a tunable parameter $\alpha$ that sets the length scale of interactions.
For $\alpha=0$, (normalized) $\alpha$-centrality takes into account direct edges only. As $\alpha$ increases, $NC(\alpha,\beta, n \to \infty)$ becomes a more global measure, taking into account ever larger network components. The expected length of a path, the radius of centrality, is $(1-\alpha)^{-1}$.

\subsection{Node Ranking}
\label{sec:ranking}
Much of the analysis done by social scientists considered local structure, i.e., the number~\cite{Wasserman:1994} and nature~\cite{Simmel, Granovetter, Burt} of an individual's ties. By focusing on local structure, however, traditional theories fail to take into account the macroscopic structure of the network. Many metrics proposed and studied over the years deal with this shortcoming, including PageRank~\cite{PageRank} and random walk centrality~\cite{Newman05}. These metrics aim to identify nodes that are `close' in some sense to other nodes in the network, and are therefore, more important.
PageRank, for example, gives the probability that a random walk initiated at node $i$ will reach $j$, while random-walk centrality computes the number of times a node $i$ will be visited by walks from all pairs of nodes in the network.

Normalized $\alpha$-centrality, $NC_i(\alpha,\beta, n \to \infty)=\sum_j{NC_{ij}(\alpha,\beta, n \to \infty)}$, also measures how `close' node $i$ is to other nodes in a network and can be used to rank the nodes accordingly. The presence of a tunable parameter turns normalized $\alpha$-centrality into a powerful tool for studying network structure and allows us to seamlessly connect the rankings produced by well-known local and global centrality metrics.
For $\alpha=0$, normalized $\alpha$-centrality takes into account local interactions that are mediated by direct edges only, and therefore, reduces to \emph{degree centrality}. As $\alpha$ increases and longer range interactions become more important, nodes that are connected by longer paths grow in importance.
For $\alpha < {1}/{\lambda_1}$, the rankings produced by normalized $\alpha$-centrality are equivalent to those produced by $\alpha$-centrality.
Also as shown in the Appendix, for symmetric matrices, as $ \alpha  \to {1}/{|\lambda_{1}|}$, normalized $\alpha$-centrality converges to \emph{eigenvector centrality}~\cite{Bonacich:2001}. The rankings no longer change as $\alpha$ increases further, since $\alpha$ has reached some fundamental length scale of the network.

\subsection{Community Detection}
\label{sec:modularity}
Girvan \& Newman~\cite{GirvanNewman04} proposed \emph{modularity} as a metric for evaluating community structure of a network.
The modularity-optimization class of community detection algorithms~\cite{Newman104,Newman204,Newman206} finds a network division that maximizes the modularity, which is defined as $Q=$ (connectivity within community)-(expected connectivity), where connectivity is measured by the density of edges.
We extend this definition to use normalized $\alpha$-centrality as the measure of network connectivity.
According to this definition, in the best division of a network, there are more weighted paths connecting nodes to others within their own community than to nodes in other communities.
Modularity can, therefore, be written as:
\begin{equation}
\label{eq: mod2}
Q(\alpha)=\sum_{ij} {[NC_{ij}(\alpha,\remove{\beta,} n \to \infty) - \overline{NC}_{ij}(\alpha,\remove{\beta,} n \to \infty)]\delta(s_i, s_j)}
\end{equation}
\noindent $NC_{ij}(\alpha, \remove{\beta,} n \to \infty)$ is given by Eq.~(\ref{eq:inf1}). Since $\beta$ factors out of modularity, without loss of generality we take $\beta=1$.
$\alpha$ can be varied from 0 to 1.
$\overline{NC}_{ij}(\alpha,\remove{\beta,} n \to \infty)$ is the expected normalized $\alpha$-centrality, and $s_i$ is the index of the community $i$ belongs to, with $\delta(s_i, s_j) = 1$ if $s_i =s_j$; otherwise, $\delta(s_i, s_j)=0$. We round the values of $NC_{ij} (\alpha, \remove{\beta,} n \to \infty)$ to the nearest integer.

To compute $\overline {NC_{ij} }(\alpha,\remove{\beta,} n \to \infty)$, we consider a graph, referred to as the null model, which has the same number of nodes and edges as the original graph, but in which the edges are placed at random. To make the derivation below more intuitive, instead of normalized $\alpha$-centrality, we talk of the number of attenuated paths. In normalized $\alpha$-centrality, the number of attenuated paths is scaled by a constant, hence the derivation below holds true. When all the nodes are placed in a single group, then  axiomatically, $Q(\alpha)=0$. Therefore $ \sum_{ij}[NC_{ij} (\alpha,\remove{\beta,} n \to \infty)- \overline {NC_{ij} }(\alpha,\remove{\beta,} n \to \infty)] =0$, and we set
$W = \sum_{ij} \overline {NC_{ij}} (\alpha,\remove{\beta,} n \to \infty)=\sum_{ij} NC_{ij}(\alpha,\remove{\beta,} n \to \infty).$
Therefore, according to the argument above, the total number of paths between nodes in the null model $\sum_{ij} \overline {NC_{ij} }(\alpha,\remove{\beta,} n \to \infty)$ is equal to the total number of paths in the original graph, $\sum_{ij} NC_{ij}(\alpha,\remove{\beta,} n \to \infty)$.  We further restrict the choice of null model to one where the expected number of paths reaching node $j$, $W_j^{in}$, is equal to the actual number of paths reaching the corresponding node in the original graph.
$
W_{j}^{in} = \sum_{i} \overline{{NC}_{ij}}(\alpha,\remove{\beta,} n \to \infty) = \sum_{i} NC_{ij}(\alpha,\remove{\beta,} n \to \infty)\,$.
Similarly, we also assume that in the null model, the expected number of paths originating at node $i$, $W_{i}^{out}$, is equal to the actual number of paths originating at the corresponding node in the original graph
$
W_{i}^{out} = \sum_{j} \overline{{NC}_{ij}}(\alpha,\remove{\beta,} n \to \infty) = \sum_{j} NC_{ij}(\alpha,\remove{\beta,} n \to \infty)
$. $W$, $W_{i}^{out}$ and $W_j^{in}$ are then rounded to the nearest integers.

Next, we reduce the original graph $G$ to a new graph $G^{\prime}$ that has the same number of nodes as $G$  and total number of edges  $W$, such that each edge has weight 1 and the number of edges between nodes $i$ and $j$ in $G^{\prime}$ is  $NC_{ij}(\alpha,\beta, n \to \infty)$. Now the expected number of paths between  $i$ and $j$  in graph $G$ could be taken as the expected number of the edges between nodes $i$ and $j$ in graph $G^{\prime}$ and the actual number of paths between nodes $i$ and $j$  in graph $G$ can be taken as the actual number of edges  between node $i$ and node $j$ in graph $G^{\prime}$. The equivalent random graph $G^{\prime \prime}$ is used to find the \emph{expected}  number of edges  from node $i$ to node $j$. In this graph  the edges are placed in random subject to constraints:
(\emph{i}) The total number of edges in $G^{\prime \prime}$ is $W$;
(\emph{ii}) The out-degree of node $i$ in $G^{\prime \prime}$ = out-degree of node $i$ in $G^{\prime} = W_{i}^{out}$;
(\emph{iii}) The in-degree of a node $j$ in graph $G^{\prime \prime}$ =in-degree of node $j$ in graph $G^{\prime} =W_{j}^{in}$.
Thus in $G^{\prime \prime}$ the  probability that an edge will emanate from a particular node  depends only on the out-degree of that node; the probability that an edge is incident on a particular node depends only on the  in-degree of that node; and the probabilities of the two nodes being the two ends of a single edge are independent of each other. In this case, the probability that an edge exists from $i$ to $j$ is given by \emph{edge in $G'$ emanates from i} $\cdot $ \emph{edge in $G'$ incident on j}=$(W_{i}^{out}/W)(W_{j}^{in}/W)$.
Since the total number of edges is $W$ in $G^{\prime \prime}$, therefore the expected number of edges between $i$ and $j$ is $W \cdot (W_{i}^{out}/W)(W_{j}^{in}/W)=\overline{{NC}_{ij}}(\alpha,\beta, n \to \infty)$, the expected the expected $\alpha$ centrality in $G$.

Once we compute $Q(\alpha)$, we have to select an algorithm to divide the network into communities that maximize $Q(\alpha)$. Brandes et al.~\cite{Brandes} have shown that the decision version of modularity maximization is NP-complete. Like others~\cite{Newman206,Leicht}, we use the leading eigenvector method to obtain an approximate solution. In this method, nodes are assigned to either of two groups based on a single eigenvector  corresponding to the largest  positive eigenvalue of the modularity matrix. This process is repeated for each group until modularity does not increase further upon division.
\remove
{
This method holds only for $\alpha < \frac{1}{\lambda_1}$. The desired  values of $\alpha$ can easily be derived from the computation of normalized $\alpha$-centrality, as that value of $\alpha$
at which  normalized $\alpha$-centrality converges.
}

\subsection{Relation to Other Centrality Measures}
We can generalize the centrality metric presented above to a notion of path-based connectivity and relate it to other centrality metrics.
Let $q=(q_{ij})$ be a $n\times n$ matrix such that $q^n_{ij}$ is the number of paths of length $n$ connecting nodes $i$ and $j$. The number of paths of length one connecting $i$ and $j$ is $q^1_{ij}=A_{ij}$; the number of paths of length two is $q^2_{ij}=(A \times A)_{ij}$, etc. The expected number of paths connecting nodes $i$ and $j$ is $E(q_{ij})$ where:
\begin{equation}
\label {eq:top1}
\mathbf{E(q)= (W_{1} \cdot q^1+ W_{2} \cdot q^2+ \ldots +W_{n} \cdot q^n+ \ldots)}\,, \nonumber
\end{equation}
\noindent where $ \bf{W_{k}}$ can be a scalar or a vector. The proximity score $E(q_{ij})$ can be used to find out how connected, or close, two nodes are.

Several path-based centrality metrics can be expressed in terms of $E(q_{ij})$, including random walk models~\cite{PageRank,Tong06,Tong08,Zhou03,Newman05}, degree centrality, Katz score~\cite{Katz}, as well as $\alpha$-centrality. In a random walk model, a particle starts a random walk at node $i$, and iteratively transitions to its neighbors with probability proportional to the corresponding edge weight. At each step, the particle returns to $i$ with some restart probability ($1-c$). The proximity score is defined as the steady-state probability $r_{i,j}$ that the particle will reach node $j$~\cite{Tong08}.
\begin{itemize}
\item If $W_{k}=c ^{k}\cdot D^{-(k)}$ where $c$ is a constant and  $D$ is  an $n \times n$ matrix with $D_{ij}=\sum_{j=1}^{n} A_{ij} $ if $i=j$ and $0$ otherwise;
then, $E(q_{ij})$ reduces to proximity score in random walk models~\cite{Tong06,Tong08}.

\item If   $W_{k}= \Pi _{j=1}^{k}\alpha_{j}$, where  the scalar $\alpha_{j}$ is the  attenuation factor along the $j$-th link in the path, then $E(q_{ij})$ reduces to $\alpha$-centrality score from $i$ to $j$ (normalizing  leads to normalized $\alpha$-centrality). For  ease of computation, we have taken  $\alpha_1=\beta$  and $\alpha_i=\alpha$,\ $\forall i\neq 1$.  $\alpha$-centrality holds for $\alpha < 1/{\lambda_1}$ ($\lambda_1$ is the largest eigenvalue of $A$). However, as shown in the Appendix, normalized $\alpha$-centrality holds for all values of $\alpha$.
\item When $\beta=\alpha$, this  in turn reduces to the Katz status score~\cite{Katz}.
\item If   $W_{k}=\alpha$, and adjacency matrix $A$ is symmetric, then as  $\alpha  \to 1/{\lambda_1}$, $E(q)$ is proportional to the inner product of the eigenvector corresponding to ${\lambda_1}$, with itself. It would lead to eigenvector centrality as shown in Appendix.
\item When  $W_{1}=1$ and $ W_{k}=0,\ \forall k>1$, then $E(q_{ij})$ is the degree centrality used in modularity-maximization approaches~\cite{Newman104}.
\end{itemize}


We have implemented~\cite{Ghosh10snakdd} a simple algorithm to compute $NC(\alpha,\beta, n \to \infty)$  using an alternative formulation of  $\alpha$-centrality:
\begin{equation}
\label{eq:inf2}
C(\alpha,\beta,n +1) =\beta A+ \alpha C(\alpha,\beta,n)A
\end{equation}
\noindent For any given value of $\alpha$, this method iteratively computes $C(\alpha,\beta, n \to \infty)$ and consequently $NC(\alpha,\beta,n \to \infty)$  until convergence.  Experimentally we have observed, that this method reaches convergence very quickly. Considering a network with $N$ nodes and $M$ links, each iteration of this algorithm (for a given value of $\alpha$) has a runtime complexity of $O(MN)$ and space complexity of $O(M)$.

 In order to study variation in this metric for $\alpha < 1/|\lambda_1|$, we must choose a step size of order $ \propto {1}/{\min(d^{out}_{max},d^{in}_{max})}$, since by the Gershgorin circle theorem, $|\lambda_1| \le \min(d^{out}_{max},d^{in}_{max})$, where $d^{out}_{max}$ and $d^{in}_{max}$ are the maximum out- and in-degree of the network respectively. Since the formulation of normalized $\alpha$-centrality is very similar to that of PageRank, similar block based strategies can be used for fast and efficient computation of both PageRank and $NC(\alpha,\beta, n \to \infty)$ ~\cite{Haveliwala:1999,Kamvar:2003}. Like PageRank, normalized $\alpha$-centrality can easily be implemented using the map-reduce paradigm~\cite{Dean:2008}, guaranteeing the scalability of this algorithm.

\section{Empirical Results}
\label{sec:results}
We apply the formalism developed above to benchmark networks studied in literature and a network extracted from the social photosharing site Flickr.

\subsection{Karate Club Network}
\begin{figure*}[t]
\begin{tabular}{ccc}
\includegraphics[width=0.32\textwidth]{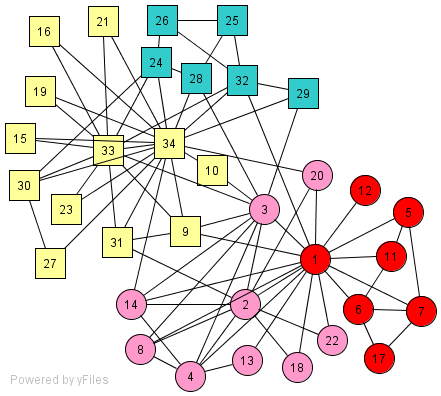} &
\includegraphics[width=0.32\textwidth]{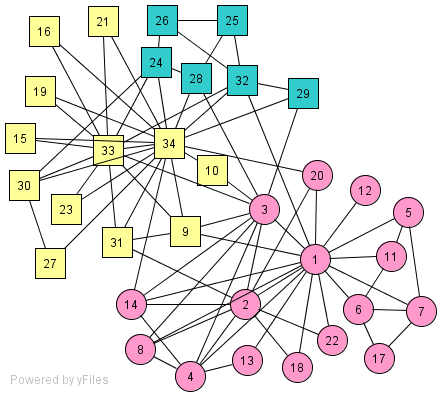} &
\includegraphics[width=0.32\textwidth]{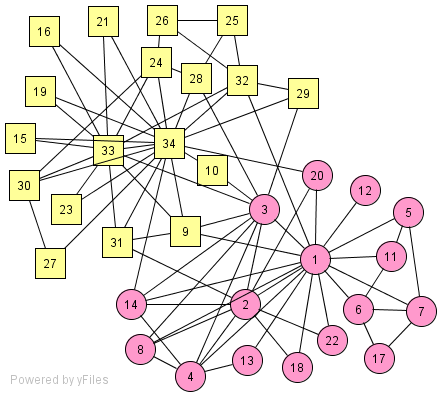} \\
(a) $ \alpha=0$ & (b) $0<\alpha<0.14$ &   (c) $\alpha \ge 0.14$
\end{tabular}
\caption{Zachary's karate club data. Circles and squares represent the two actual factions, while colors stand for discovered communities as the strength of ties increases: (a) $ \alpha=0$, (b) $0<\alpha<0.14$, (c) $\alpha \ge 0.14$
}
 \label{fig:karate}
\end{figure*}

First, we study the friendship network of Zachary's karate club~\cite{Zachary} shown in Figure~\ref{fig:karate}. During the course of the study, a disagreement developed between the administrator and the club's instructor, resulting in the division of the club into two factions, represented by circles and  squares in \figref{fig:karate}.
We find community division of this network for $ 0 \leq \alpha \le 0.1487$ (maximum $\alpha$ is given by reciprocal of the largest eigenvalue of the adjacency matrix). The first bisection of the network results in two communities, regardless of the value of $\alpha$, which are identical to the two factions observed by Zachary.  However, when the algorithm runs to termination (no more bisections are possible), different groups are found for different values of $\alpha$. For $\alpha=0$, the method reduces to edge-based modularity maximization~\cite{Newman204} and leads to four groups~\cite{Duch05,Fortunato10} (\figref{fig:karate}(a)). For  $0<\alpha<0.14$  it discovers three groups (\figref{fig:karate}(b)), and for $\alpha > 0.14$,  two groups that are identical to the factions found by Zachary (\figref{fig:karate}(c)). Thus, increasing $\alpha$ allows local groups to merge into more global communities.

\begin{figure}[th]
\begin{center}
   \includegraphics[width=\columnwidth]{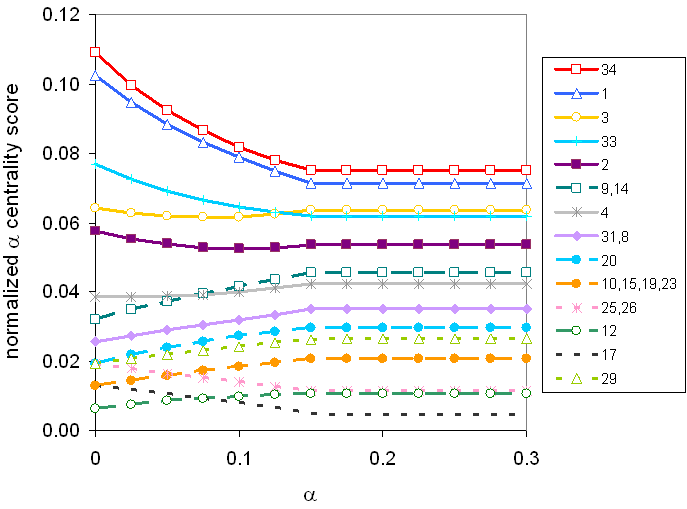}
\caption{Centrality scores of Zachary club members vs.\ $\alpha$.}
\label{fig:rankings}
\end{center}
\end{figure}

Figure~\ref{fig:rankings} shows how the normalized $\alpha$-centrality scores of nodes change with $\alpha$. For $\alpha=0$, normalized $\alpha$-centrality reproduces the rankings given by degree centrality. As we show in the appendix, the final rankings produced by normalized $\alpha$-centrality for this symmetric matrix are the same as those given by the eigenvector centrality.  This can be confirmed by their values in  Table \ref{tbl:kl}. Varying $\alpha$ allows us to smoothly transition from a local to a global measure of centrality.

Nodes 34 and 1 have the highest centrality scores, especially at lower $\alpha$ values. These are the \emph{leaders} of their communities. It was the disagreement between these nodes, the club administrator (node 1) and instructor (node 34), that led to the club's division.  Nodes 33 and 2 also have high centrality and hold leadership positions. All these nodes are also scored highly by betweenness centrality and PageRank. Note that centrality scores of these nodes decrease with $\alpha$, indicating that they are far more important locally than globally.

A node may also have high centrality if it is connected to many nodes from different communities. Such nodes, which \emph{bridge} communities, are crucially important to maintaining cohesiveness and facilitating communication flow in both human~\cite{Simmel,Granovetter} and animal~\cite{Lusseau04} groups. We can identify these nodes because their normalized $\alpha$-centrality increases with $\alpha$, i.e., they become more important as longer paths become more important. Centrality of nodes 3, 14, 9, 31, 8, 20, 10, etc., increases with $\alpha$ from moderate to relatively high values. While most of these nodes are directly connected to both communities, some are only indirectly connected by longer paths. Betweenness centrality of these nodes is low, but non-zero.

Nodes 25, 26 and 17 have low centrality which decreases with $\alpha$. These are \emph{peripheral} members. Betweenness centrality of 17 is zero, as expected, but 25 and 26 have scores similar to 31. PageRank scores of these peripheral nodes are higher than nodes 21, 22, 23, which are connected to central nodes, and comparable to scores of the bridging nodes 20 and 31. While both betweenness centrality and PageRank correctly pick out leaders, they do not distinguish between locally and globally connected nodes.


\begin{table}[htdp]
\caption{Comparison of eigenvector centrality and converged normalized  $\alpha$-centrality for Zachary's karate club network.}
\begin{center}
\begin{tabular}{|c|c|c||c|c|c|}
\hline
$node$	&	$\alpha$-cen	&	$eigenvector$ & $node$	&	$\alpha$-cen	&	$eigenvector$	\\ \hline

34	&	0.075		&	0.3734	&	15	&	0.0204	&	0.1014\\ \hline
1	&	0.0714	&	0.3555	&	16	&	0.0204	&	0.1014\\ \hline
3	&	0.0637	&	0.3172	&	19	&	0.0204	&	0.1014\\ \hline
33	&	0.062		&	0.3086	&	21	&	0.0204	&	0.1014\\ \hline
2	&	0.0534	&	0.266		&	23	&	0.0204	&	0.1014\\ \hline
9	&	0.0457	&	0.2274	&	18	&	0.0186	&	0.0924\\ \hline
14	&	0.0455	&	0.2265	&	22	&	0.0186	&	0.0924\\ \hline
4	&	0.0424	&	0.2112	&	13	&	0.0169	&	0.0843\\ \hline
32	&	0.0384	&	0.191		&	6	&	0.016		&	0.0795\\ \hline
31	&	0.0351	&	0.1748	&	7	&	0.016		&	0.0795\\ \hline
8	&	0.0343	&	0.171		&	5	&	0.0153	&	0.076\\ \hline
24	&	0.0302	&	0.1501	&	11	&	0.0153	&	0.076\\ \hline
20	&	0.0297	&	0.1479	&	27	&	0.0152	&	0.0756\\ \hline
30	&	0.0271	&	0.135		&	26	&	0.0119	&	0.0592\\ \hline
28	&	0.0268	&	0.1335	&	25	&	0.0115	&	0.0571\\ \hline
29	&	0.0263	&	0.1311	&	12	&	0.0106	&	0.0529\\ \hline
10	&	0.0206	&	0.1027	&	17	&	0.00475	&	0.0236\\ \hline	
\remove{
\begin{table}[htdp]
\caption{Comparison of Eigen Vector centrality and Normalized  $\alpha$ centrality at convergence for Zachary's Karate club}
\begin{center}
\begin{tabular}{|c|c|c|}
\hline
node	&	alpha	&	eigenvector	\\ \hline
34	&	0.0750	&	0.3734	\\ \hline
1	&	0.0714	&	0.3555	\\ \hline
3	&	0.0637	&	0.3172	\\ \hline
33	&	0.0620	&	0.3086	\\ \hline
2	&	0.0534	&	0.2660	\\ \hline
9	&	0.0457	&	0.2274	\\ \hline
14	&	0.0455	&	0.2265	\\ \hline
4	&	0.0424	&	0.2112	\\ \hline
32	&	0.0384	&	0.1910	\\ \hline
31	&	0.0351	&	0.1748	\\ \hline
8	&	0.0343	&	0.1710	\\ \hline
24	&	0.0302	&	0.1501	\\ \hline
20	&	0.0297	&	0.1479	\\ \hline
30	&	0.0271	&	0.1350	\\ \hline
28	&	0.0268	&	0.1335	\\ \hline
29	&	0.0263	&	0.1311	\\ \hline
10	&	0.0206	&	0.1027	\\ \hline
15	&	0.0204	&	0.1014	\\ \hline
16	&	0.0204	&	0.1014	\\ \hline
19	&	0.0204	&	0.1014	\\ \hline
21	&	0.0204	&	0.1014	\\ \hline
23	&	0.0204	&	0.1014	\\ \hline
18	&	0.0186	&	0.0924	\\ \hline
22	&	0.0186	&	0.0924	\\ \hline
13	&	0.0169	&	0.0843	\\ \hline
6	&	0.0160	&	0.0795	\\ \hline
7	&	0.0160	&	0.0795	\\ \hline
5	&	0.0153	&	0.0760	\\ \hline
11	&	0.0153	&	0.0760	\\ \hline
27	&	0.0152	&	0.0756	\\ \hline
26	&	0.0119	&	0.0592	\\ \hline
25	&	0.0115	&	0.0571	\\ \hline
12	&	0.0106	&	0.0529	\\ \hline
17	&	0.0048	&	0.0236	\\ \hline
}
\end{tabular}
\end{center}
\label{tbl:kl}
\end{table}

\begin{figure}[th]
\begin{center}
   \includegraphics[width=\columnwidth]{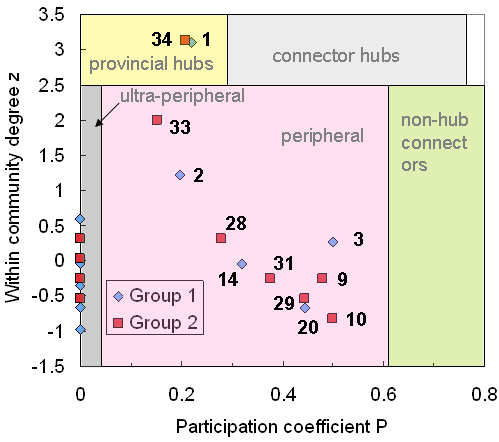}
\caption{Classification of karate club nodes according to the roles scheme proposed by Guimera et al.~\protect\cite{Guimera05}: (i) non-hubs ($z < 2.5$) are divided into   \emph{ultra-peripheral}, \emph{peripheral}, and \emph{connector} nodes (\emph{kinless} nodes whose links are homogeneously distributed among all communities are not shown); (ii) hubs ($z \ge 2.5$) are subdivided into \emph{provincial} (majority of link within their own community), \emph{connector hubs} (many links to other communities). \emph{Global hubs} whose links are homogeneously distributed among all communities are not shown.}
\label{fig:roles}
\end{center}
\end{figure}

Guimera and collaborators~\cite{Guimera05} proposed a role-based description of complex networks as an alternative to the `average description' approach, which characterizes network structure in terms of average degree or degree distribution. They define a role in terms of the relative within-community degree $z$ (which measures how well the node is connected to other nodes in its community) and participation coefficient $P$ (which measures how well the node is connected to nodes in other communities). They propose a heuristic classification scheme to assign roles to nodes based on where they fall in the $z$--$P$ plane and find similar patterns of role-to-role connectivity among networks with similar functional needs and growth mechanisms~\cite{Guimera07}.

Figure~\ref{fig:roles} shows the positions of nodes in the karate club network in the $z$--$P$ plane. Colored regions demarcate the boundaries of different roles according to Guimera et al.'s classification scheme. Nodes separate into provincial hubs (34, 1), peripheral (33, 2, 28, 14, 31, 29, 20, 3, 9, 10) and ultra-peripheral nodes (rest of the nodes). No special role is assigned to the bridging nodes, such as 9. Even if the boundary of non-hub connectors is shifted to slightly less than $P=0.5$ in order to identify nodes 3, 9, 10 as serving a special role, the method would still miss node 14, whose position in the network is very similar to node 9. This is because the method takes into account direct links only, rather than complete connectivity between nodes. The method also requires one to first identify communities in the network, which is a very computationally expensive procedure for large networks. Our method, on the other hand, uses only matrix multiplication and provides a computationally efficient and scalable way to identify network structure.

\subsection{Florentine Families}
Padgett~\cite{Padgett} studied the structure of political and business relationships among the elite families of Renaissance Florence. The two rival factions during this period were the \emph{oligarchs}, composed of the patrician families, and the \emph{Mediceans}, who formed close ties with the newly powerful businessmen, or the ``new men.'' Before the rise of Medicis, the oligarchs dominated Florentine politics and economics and cemented their power through marriage. They were less willing to enter into business relationships with the ``new men.'' The Medicis, on the other hand, consolidated their power through business and marriage relationships. Scholars have studied the business and marriage networks of Renaissance Florence to explain the outcomes of the power struggles between the factions and the rise of the Medici family during this important period of Western European history.

\begin{figure*}[tbh]
\begin{tabular}{ccc}
     \includegraphics[width=0.32\textwidth]{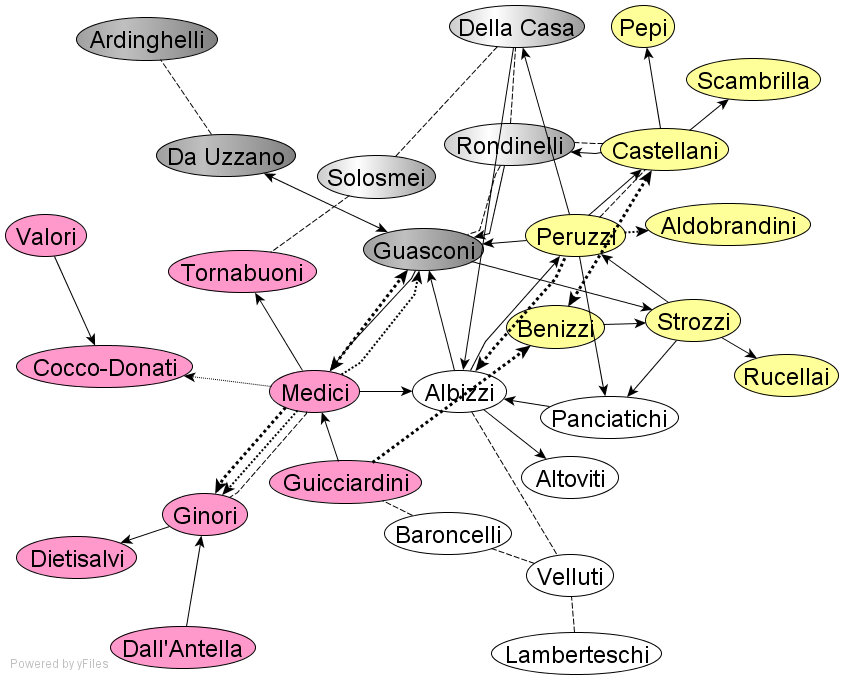} &
 \includegraphics[width=0.32\textwidth]{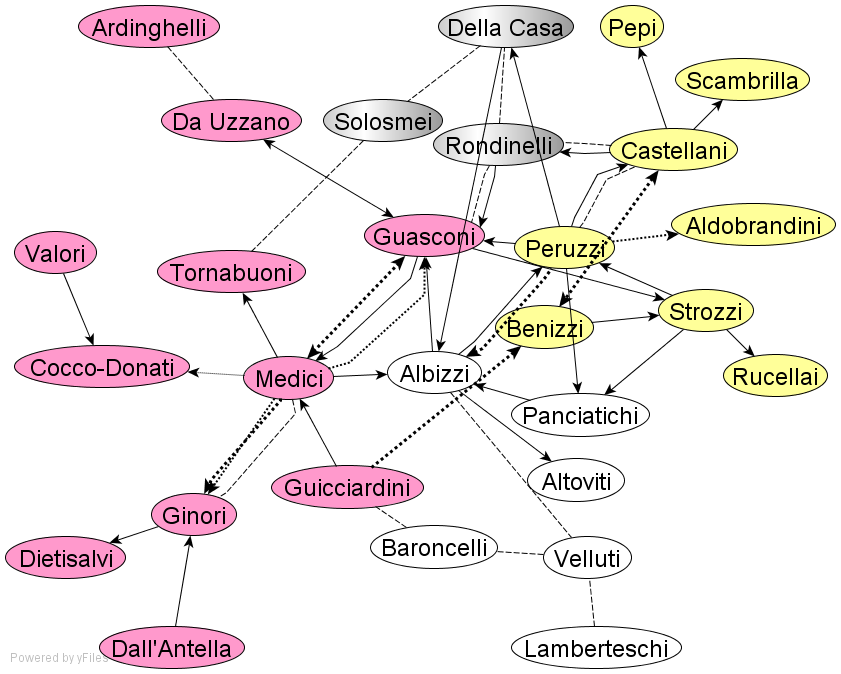} &
   \includegraphics[width=0.32\textwidth]{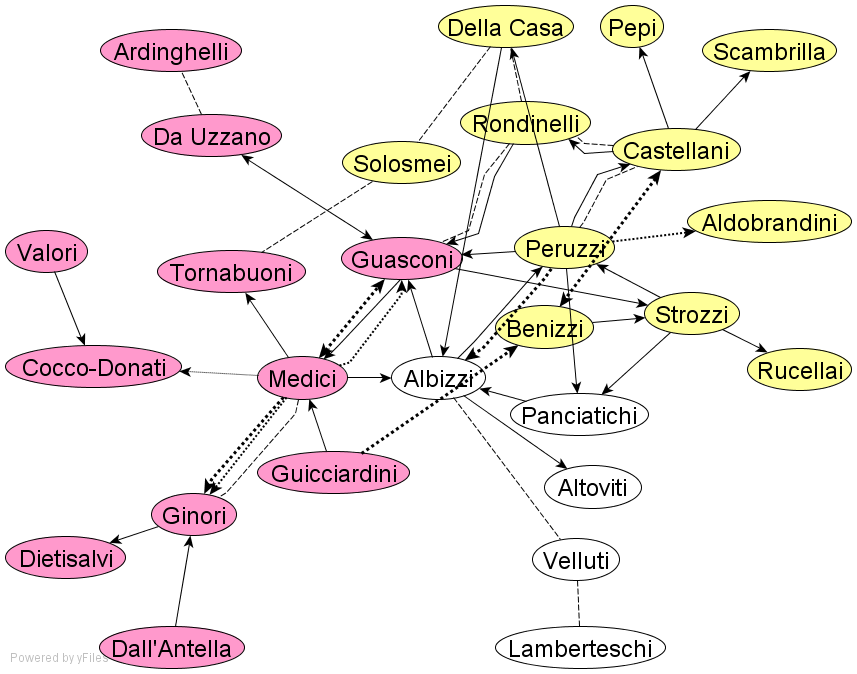}
 \\
     (a) $0\le \alpha \le 0.05$ & (b) $0.05<\alpha \le 0.1$ &   (c) $0.1 < \alpha \le 0.25$
\end{tabular}
\caption{The groups in the Florentine family data set: (a)$ 0\le \alpha \le 0.05$ (b) $0.05<\alpha \le 0.1$ (c)$0.1 < \alpha \le 0.25$}
\label{fig:fam_groups}
\end{figure*}


We applied our community  detection algorithm to the heterogenous network containing both marriage and business ties shown in Figure~\ref{fig:fam_groups}. The marriage ties are shown by straight lines. The dashed lines show the different business relations.
The marriage ties are asymmetric, with the wife-giving family being considered superior to the wife-receiving family.
All relations are weighted equally. We symmetrized the resulting adjacency matrix by adding it to its transpose. We studied community division of this network for values of $\alpha$ in the range $0 \leq \alpha \le 0.25$, since $0.25 < 1/{\lambda_1} <0.26$. For $ 0\leq \alpha \le 0.05$ we found seven  distinct groups, shown in Fig.~\ref{fig:fam_groups}(a). Two of these are small and disconnected from the rest of the network. The first of these groups is composed of Guadigni, Fioravanti and Bischeri families and the other of  Orlandini and Davazati families.
The three largest groups within the connected component are:  $(1)$ families aligned with Medici, $(2)$ families aligned with the oligarchs, such as Strozzi, Peruzzi,  and
$(3)$ mostly oligarch families with split loyalties, like the Alibizzi.
For lower values of $\alpha$, the oligarchs are split into two groups (Fig.~\ref{fig:fam_groups}(a), (b)).  The rift  within the oligarchs detected by our algorithm  is corroborated by historic events. When a lottery randomly produced too many Medici officeholders in the Signoria (1433), Rinaldo Albizzi, the titular head of the oligarchs,  sent out a word to assemble troops in order to forcibly seize Signoria from the Medicis. However, his repeated efforts to assemble troops (especially from Palla Strozzi) were frustrated by other supporters' changing their minds and drifting away~\cite{Padgett}, indicating factional split within the oligarchs.
In contrast, Medicis could immediately and effectively mobilize their supporters, as the result of which no military action ensued and Cosimo Medici took over the budding Florentine state.

As we increase the scale of interactions by increasing $\alpha$, the five groups within the connected component gradually coalesce into three distinct, as shown in Fig.~\ref{fig:fam_groups}. First, the group comprising of Guasconi, Da-Uzzano and Ardinghelli integrate with the Medicis (Fig.~\ref{fig:fam_groups}(b)). When $\alpha$ is  further increased, the group comprising of Rondinelli, Solosmei and Della Casa integrate with the oligarchs (Fig.~\ref{fig:fam_groups}(c)).


\begin{figure}[htbp]
\begin{center}
   \includegraphics[width=3.5in]{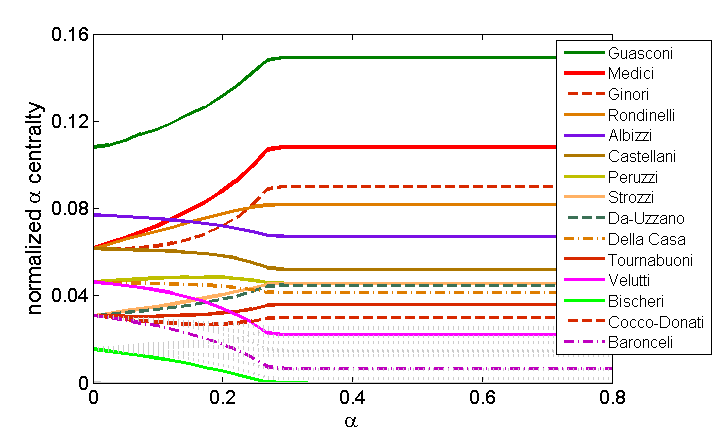}
\caption{Normalized $\alpha$-centrality scores of the families within the business-marriage network. Some of the nodes are identified, with the rest shown in grey.}
\label{fig:all}
\end{center}
\end{figure}

\remove{
\begin{figure}[htbp]
\begin{center}
   \includegraphics[width=3.5in]{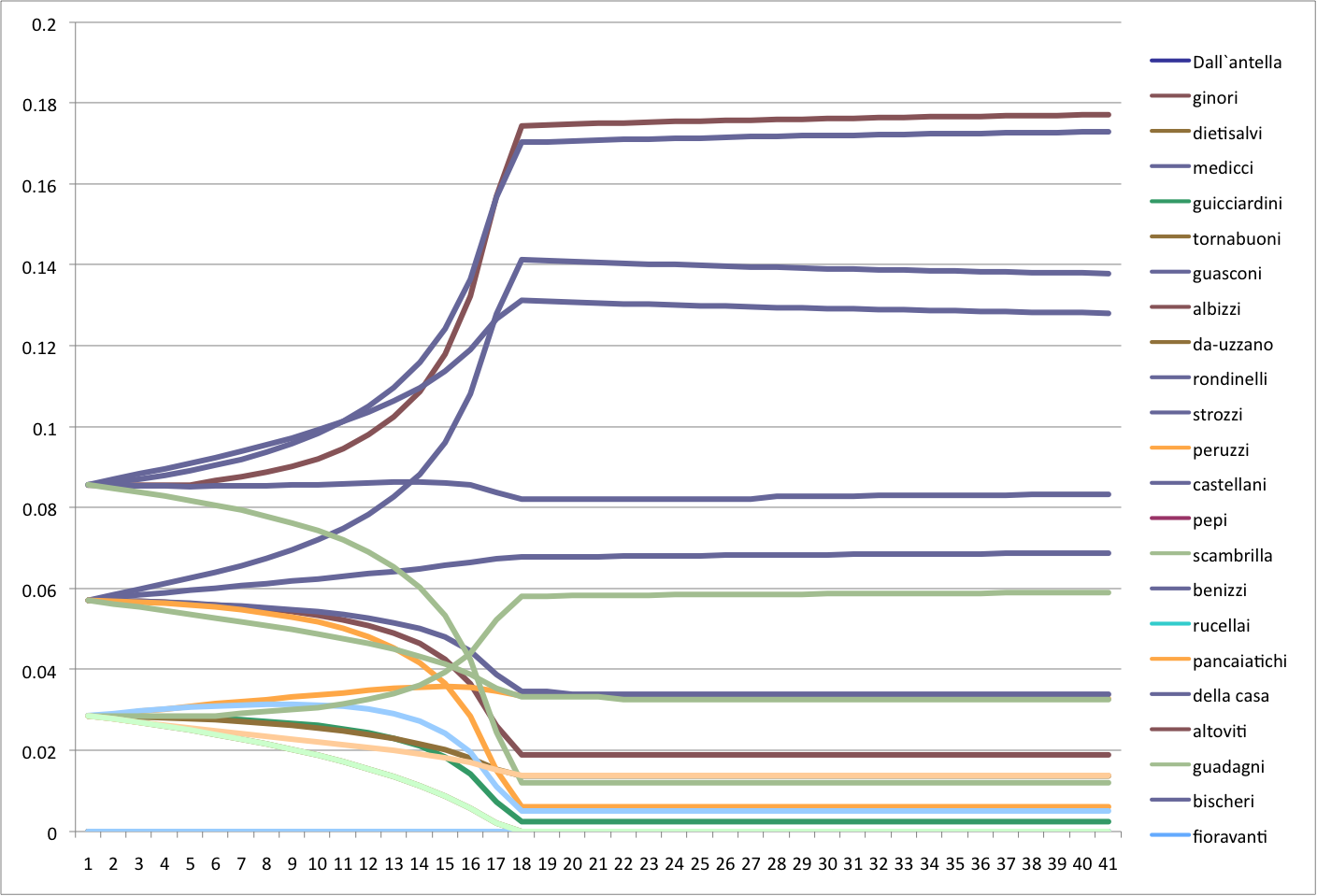}
\caption{$\alpha$ centrality scores of the families within the network comprising of only business connections.}
\label{fig:business}
\end{center}
\end{figure}
\begin{figure}[htbp]
\begin{center}
   \includegraphics[width=3.5in]{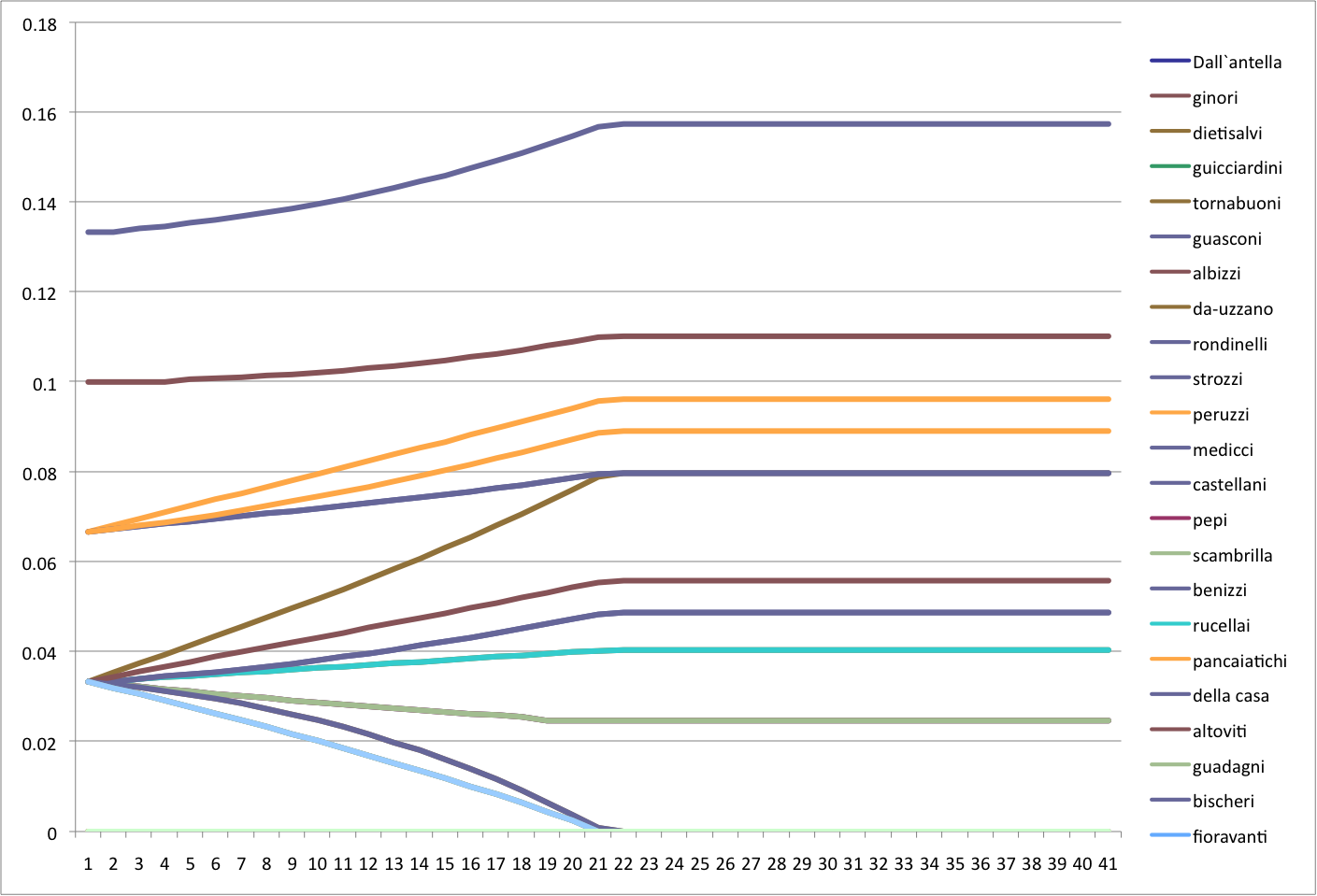}
\caption{$\alpha$ centrality scores of the families within the network comprising of only business connections.}
\label{fig:marriage}
\end{center}
\end{figure}
}
Figure~\ref{fig:all} shows how the normalized $\alpha$-centrality scores of the families in the heterogeneous business-marriage network change with $\alpha$.
Guasconi family has the highest centrality score across all values of $\alpha$. This is not surprising given this family's central position in the network {bridging} the oligarchs and the Mediceans. This observation is corroborated by the findings that cross-pressurised (by the Mediceans and the oligarchs) Guasconis  were split in their partisan loyalties~\cite{Padgett}.
Similarly, the Medicis, who were able to expertly exploit both business and marriage connections, increase in importance as $\alpha$ increases.
On the other hand, the oligarchs such as Strozzi and Peruzzi families, who were patrician to the core and had few business relations outside of their faction, see their centrality decrease with $\alpha$.
Thus, the historic ascendance of the Mediceans can be observed already in the business and marriage networks they created.



Although the heterogenous network of Florentine families is not a symmetric network, since it contains asymmetric marriage relations, we find that rankings produced by normalized $\alpha$-centrality for $\alpha > 1/{|\lambda_1|}$ is well correlated with those produced by eigenvector centrality. However, as observed by Bonacich~\cite{Bonacich:2001}, for asymmetric the marriage network there are important differences between the rankings of eigenvector centrality and normalized $\alpha$-centrality. For instance,  the eigenvector centrality scores of Bischeri, Guadigni and Orlandini  are zero, even though they were wife-giving families. The reason for this is their segregation from the large connected component. This anomaly is corrected by the normalized $\alpha$-centrality, though the centrality scores for Bischeri, Guadigni and Orlandini  are very small.

\remove{
\begin{table}[htdp]
\caption{Comparison of eigenvector centrality and normalized  $\alpha$-centrality scores at convergence of the Florentine families.}
\begin{center}
\begin{tabular}{|c|c|c|}
\hline
$family$	&	$eigenvector$	&	\emph{normalized $\alpha$-centrality}	\\ \hline
Guasconi	&	0.564354	&	0.1490	\\ \hline
Medicci	&	0.409375	&	0.108	\\ \hline
Ginori	&	0.340562	&	0.0899	\\ \hline
Rondinelli	&	0.309436	&	0.0817	\\ \hline
Albizzi	&	0.254793	&	0.0673	\\ \hline
Castellani	&	0.197605	&	0.0522	\\ \hline
Peruzzi	&	0.173211	&	0.0457	\\ \hline
Strozzi	&	0.172231	&	0.0455	\\ \hline
Da-Uzzano	&	0.169533	&	0.0448	\\ \hline
Della Casa	&	0.156317	&	0.0413	\\ \hline
Tornabuoni	&	0.135999	&	0.0359	\\ \hline
Cocco-Donati	&	0.113521	&	0.0300	\\ \hline
Pancaiatichi	&	0.095792	&	0.0253	\\ \hline
Dietisalvi	&	0.094439	&	0.0249	\\ \hline
Velluti	&	0.084133	&	0.0222	\\ \hline
Solosmei	&	0.08106	&	0.0214	\\ \hline
Altoviti	&	0.070655	&	0.0187	\\ \hline
Benizzi	&	0.05674	&	0.0150	\\ \hline
Pepi	&	0.054796	&	0.0145	\\ \hline
Scambrilla	&	0.054796	&	0.0145	\\ \hline
Aldobrandini	&	0.048032	&	0.0127	\\ \hline
Rucellai	&	0.04776	&	0.0126	\\ \hline
Ardinghelli	&	0.047012	&	0.0124	\\ \hline
Baroncelli	&	0.025274	&	0.0067	\\ \hline
Lamberteschi	&	0.02333	&	0.0062	\\ \hline
Guicciardini	&	0.007008	&	0.0019	\\ \hline
Bischeri	&	0	&	0	\\ \hline
Guadagni	&	0	&	0	\\ \hline
Orlandini	&	0	&	0	\\ \hline
Dall`Antella	&	0	&	0	\\ \hline
Fioravanti	&	0	&	0	\\ \hline
Valori	&	0	&	0	\\ \hline
Davanzati	&	0	&	0	\\ \hline

\end{tabular}
\end{center}
\label{tbl:fl}
\end{table}
}

\subsection{Other Real-World Networks}
\label{sec:other}
In addition to the social networks described above, we evaluated the performance of our community division algorithm on three other real-world networks:
the US College football 
and the political books networks, as well as the social network retrieved from the social photosharing site Flickr. We were not able to evaluate rankings due to the lack of ground truth for these data sets.
The first network represents the schedule of Division 1 games for the 2001 season where the nodes represent teams and the edges represent the regular season games between teams~\cite{GirvanNewman02}. The teams are divided into conferences containing 8 to 12 teams each. Games are more frequent between members of the same conference, though inter-conference games also take place. This leads to an intuition, that the natural communities may be larger than conferences.

The political books network represents books about US politics sold by the online bookseller Amazon.\footnote{\texttt{http://www.orgnet.com/}} Edges represent frequent co-purchasing  by the same buyers, as indicated by the ``customers who bought this book also bought these other books'' feature of Amazon.  The nodes were labeled \emph{liberal}, \emph{neutral}, or \emph{conservative} by Mark Newman on a reading their descriptions and reviews on Amazon\footnote{\texttt{http://www-personal.umich.edu/$\sim$mejn/netdata/}}. We take these labels as communities.

To collect the final data set, we sampled Flickr's social network by identifying roughly 2000 users interested in one of three topics: \emph{portraiture}, \emph{wildlife}, and \emph{technology}.
We used the Flckr API to perform a tag search using relevant keywords to retrieve 500 `most interesting' images for each topic and extracted the names of users who uploaded these images.\footnote{The keywords used for image search were (a) \textsf{newborn} for the \emph{portraiture} topic, (b) \textsf{tiger} and \textsf{beetle} for the \emph{wildlife} topic, and (c) \textsf{apple} for the \emph{technology} topic.}
Further, we identified four  users (eight for the \emph{wildlife} topic) who were interested in each topic by studying their profiles, specifically  group membership and user's tags.
Groups such as ``Big Cats'', ``Zoo'', ``The Wildlife Photography'', etc.  pointed to user's interest in \emph{wildlife}. In addition, tags that users attached to their images could also help identify their interests. Users who used \textsf{nature} and \textsf{macro} tags were probably interested \emph{wildlife} rather than \emph{technology}. Similarly, users interested in human, rather than animal, \emph{portraiture} tagged their images with \textsf{baby} and \textsf{family}.
We then used Flickr API to retrieve these users' contacts, as well as their contacts' contacts, and labeled all by the topic through which they were discovered.
We reduced this network to an undirected network of mutual contacts only, resulting in a network of $5747$ users, with $1620$, $1337$ and $2790$ users labeled  \emph{technology}, \emph{portraiture} and \emph{wildlife} respectively. Although we did not verify that all the users were interested in the topics they were labeled with, we use these `soft' labels to evaluate the discovered communities.

We use \emph{purity} to evaluate the quality of discovered communities. We define purity as the fraction of all pairs of objects in the same community that are assigned to the same group by the algorithm. This is a simplified version of the Wallace criterion~\cite{Wallace83} for evaluating performance of clustering algorithms.

\begin{table}[htdp]
\caption{The number and purity of communities discovered at different values of $\alpha$}
\begin{center}
\begin{tabular}{|c|c|c||c|c|c||c|c|c|}
\hline
\multicolumn{3}{|c||}{\emph{karate club}} &
\multicolumn{3}{c||}{\emph{football}} &
\multicolumn{3}{c|}{\emph{flickr}} \\ \hline
$\alpha$ & grps & Pu & $\alpha$ & grps & Pu & $\alpha$ & grps & Pu\\ \hline
0.00 	& 4 &	0.505 	& 0.00 & 8 	& 0.715	& 0.000    & 4 &	0.501 \\
0.12 	& 3 & 	0.736 	&  0.02 & 8	& 0.723	&  0.001 & 3 &	0.565\\
0.14 	& 2 & 	1.000	& 0.04 & 8	& 0.723	&  0.002 & 3 &	 0.567\\ \cline{1-3}
\multicolumn{3}{|c||}{\emph{florentine}} & 0.06 & 7	& 0.723	&  0.003 & 3 &	0.567\\  \cline{1-3}
0.00    & 7 & 0.34      & 0.08 & 7	& 0.723	&  0.004 & 3 &	0.567\\
0.05    & 6 & 0.34      & 0.10 & 7	& 0.791	&  0.005 & 3 &	0.568\\
0.10    & 5 & 0.42      & 0.12 & 6	& 0.803	&  0.006 & 3 &	0.570\\\cline{1-3}
\multicolumn{3}{|c||}{\emph{political books}} &  0.14 & 6	& 0.813	&  0.007 & 3 &	0.571\\\cline{1-3}
0.00    & 4 &	0.633	& 0.16 & 6 	& 0.813 	& 0.008 & 3 &	0.572\\
0.04 	& 3 & 	0.805 	& 0.18 & 4	& 0.862	&  0.009 & 3 & 	0.574\\
0.08 	& 2 & 	0.917 	&       &   &       &       &   & \\ \hline
\end{tabular}
\end{center}
\label{tbl:mi}
\end{table}

The number and purity of the communities found in networks as a function of the parameter $\alpha$ are shown in Table~\ref{tbl:mi}.  The case $\alpha=0$ corresponds to edge-based modularity method.  As $\alpha$ increases, the number of groups discovered in all networks goes down, while their purity increases. This is consistent with our hypothesis that using smaller values of $\alpha$ allows us to identify more local network structure, while larger values of $\alpha$ lead to more global structure. In the Karate club network, for example, at $\alpha=0$, there are four small communities, as shown in Fig.~\ref{fig:karate}(a). These local communities coalesce into two large groups as $\alpha$ increases (Fig.~\ref{fig:karate}(c)), which are identical to the groups identified by Zachary~\cite{Zachary}.

To evaluate the communities discovered in the Florentine families network, we use the tight constraint of party loyalty.  Hence the families could be either Medicean, oligarch, or have split loyalties. We note that this is a very conservative evaluation criterion, since there were factions present within the parties themselves~\cite{Padgett}. Since families with split loyalties would be correctly classified as belonging to either of the two parties, we remove them from purity calculation, focusing instead on identifying community of party loyalists only. Purity is further reduced by the presence of the two isolated groups. However, purity of discovered communities increases with $\alpha$.
The small local communities found at lower value of $\alpha$ could indicate factions within parties.


\section{Related Work}
\label{sec:related}
A variety of metrics have been proposed to measure node's centrality in a network~\cite{Katz,Hubbell:1965,Bonacich:1972,Freeman,Wasserman:1994,PageRank,Bonacich:2001,Newman05}, yet few studies systematically evaluated their performance on real-world networks.
Liben-Nowell and Kleinberg~\cite{Liben07} compared the performance of several commonly used centrality metrics on the link prediction task and found Katz score~\cite{Katz} to be the most effective measure for this task, outperforming PageRank~\cite{PageRank} and its variants.
The $\alpha$-centrality metric modifies the Katz score by introducing a parameter $\alpha$, that gives a weight to indirect links and also sets the length scale of interactions in the network. We showed recently~\cite{Ghosh10snakdd} that normalized $\alpha$-centrality outperforms other centrality metrics on the task of predicting influential nodes in an online social network.

Guimera and collaborators~\cite{Guimera05,Guimera07} proposed role-based description of complex networks. They define a role in terms of the relative within-community degree $z$ (which measures how well the node is connected to other nodes in its community) and participation coefficient $P$ (which measures how well the node is connected to nodes in other communities). They proposed a heuristic classification scheme based on where the nodes lie in the $z$--$P$ plane. This classification scheme is similar to the local vs. globally-connected distinction we are making, with connector nodes being more globally connected nodes while provincial hubs and peripheral nodes are more locally connected. Role-based analysis requires community decomposition of the network to be performed first. This is a computationally expensive procedure for most real-world networks. Our approach, on the other hand, allows us to differentiate between roles of nodes in a more computationally efficient way.

Community detection is another active area in networks research (see \cite{Fortunato10} for a comprehensive review). Like us, Arenas et al.~\cite{Arenas} have generalized modularity to find correlations between nodes that go beyond nearest neighbors.
Their approach relies on the presence of motifs~\cite{Milo02,Milo04}, i.e., connected subgraphs such as cycles, to identify communities within a network.
For example, higher than expected density of triangles implies presence of a community, and a triangle modularity may be defined to identify it. The motif-based modularity uses the size of the motif to impose a limit on the proximity of neighbors.
Our method, on the other hand, imposes no such limit. The measure of global correlation computed using $\alpha$-centrality  is equal to the weighted average of correlations  for motifs of different sizes. Our method enables us to easily calculate this complex term.

\section{Conclusion}
In this paper, we introduced normalized $\alpha$-centrality as a metric to study network structure. Like the original $\alpha$-centrality~\cite{Bonacich:2001} on which it is based, this metric measures the number of paths that exist between nodes in a network, attenuated by their length with the attenuation parameter $\alpha$. This parameter sets the length scale of the interaction. When $\alpha=0$, the centrality metric takes into account direct edges only and is equivalent to degree centrality. As $\alpha$ increases, the metric takes into consideration more distant network interactions, becoming a more global measure. Normalized $\alpha$-centrality allows us to smoothly interpolate between local metrics, such as degree centrality, and global metrics, such as eigenvector centrality~\cite{Bonacich:2001}. Unlike the original $\alpha$-centrality, which bounds  $\alpha$ to be less than the reciprocal of the largest eigenvalue of the adjacency metric of the network, normalized $\alpha$-centrality sets no such limit.

We used normalized $\alpha$-centrality to study the structure of networks, specifically, identify important nodes and communities within the network. We extended the modularity maximization class of algorithms~\cite{GirvanNewman02} to use (normalized) $\alpha$-centrality, rather than edge density, as a measure of network connectivity. For small values of $\alpha$ smaller, more locally connected communities emerge, while for larger values of $\alpha$, we observe larger globally connected communities. We also used this metric to rank nodes in a network.  By studying changes in rankings that occur when parameter $\alpha$ is varied, we were able to identify locally important `leaders' and globally important `bridges' or `brokers' that facilitate communication between different communities. We applied this approach to  benchmark networks studied in literature and found that it results in network division in close agreement with the ground truth. We can easily extend this definition to multi-modal networks that link entities of different types, and use approach described in this paper to study the structure of such networks~\cite{Ghosh09socialcom}.

\begin{acknowledgments}
This work is supported in part by the NSF under award 0915678 and in part by AFOSR.
\end{acknowledgments}

\appendix

\section{Proofs of Convergence}
\label{Appendix}
In this section we prove some properties of the normalized $\alpha$-centrality metric proposed in this paper. The \emph{$\alpha$-centrality matrix} $C(\alpha, \beta,n)$  $\forall \alpha \in [0,1]$ is defined as:
\begin{eqnarray}
C(\alpha,\beta,n) &=& \beta A(I+ \alpha A+\alpha^2 A^2+\cdots+\alpha^{n} A^{n}) \nonumber  \\
&=&{\displaystyle \beta A\sum_{k=0}^n} \alpha^{k}A^{k}
\label{eq:bc_eq}
\end{eqnarray}
\noindent The \emph{normalized  $\alpha$-centrality matrix} is then given by:
\begin{equation}
NC(\alpha, \beta, n)=\frac{C(\alpha, \beta, n)}{\displaystyle \sum_{i,j} {(C(\alpha, \beta, n))}_{ij}}
\end{equation}
\noindent The  normalized $\alpha$-centrality vector is $NC_i(\beta ,\alpha,n\to \infty) =eNC(\alpha,\beta, n \to \infty)$ where $e$ is a $1\times N$ unit vector and $N$ is the number of nodes in the network.

 If $\lambda$ is an \emph{eigenvalue} of $A$, then
\begin{eqnarray}
(I-\frac{1}{\lambda}A)x=0
\end{eqnarray}
Invertibility of $(I-\frac{1}{\lambda}A)$ would lead to the trivial solution of eigenvector $x$ ($x=0$).
Hence for computation of eigenvalues and eigenvectors, we require that no inverse of $(I-\frac{1}{\lambda}A)$  should exist, i.e.
\begin{equation}
Det(I - \frac{1}{\lambda} A)= 0.
\label{eq:char_eq}
\end{equation}
Equation~(\ref{eq:char_eq}) is called the \emph{characteristic equation}, solving which gives the \emph{eigenvalues} and \emph{eigenvectors} of the adjacency matrix $A$.

The adjacency matrix $A$ can be written as:
\begin{equation}
A=X\Lambda X^{-1} = \sum_{i=1}^{N} \lambda_{i}Y_{i}
\label{eq:eigen}
\end{equation}
where $X$ is a matrix whose columns are the eigenvectors of $A$, and
$\Lambda$ is a diagonal matrix whose diagonal elements are the eigenvalues of $A$, $\Lambda_{ii}=\lambda_{i}$, arranged according to the ordering of the eigenvectors in $X$.
Without loss of generality we assume that $\lambda_{1}> \lambda_{2} >\cdots >\lambda_{N}$.
The matrices $Y_{i}$ can be determined from the product
\begin{equation}
Y_{i}=X {Z}_{i}X^{-1}
\label{eq:z}
\end{equation}
where $Z_{i}$ is the \emph{selection matrix} having zeros everywhere except for element ${(Z_i)}_{ii}=1$ ~\cite{Gebali:2008}.

Adjacency matrix $A$ raised to the power $k$ is then given by
\begin{equation}
A^{k}=X{\Lambda}^{k} X^{-1} =\sum_{i=1}^{N} {\lambda_{i}^{k}}Y_{i}
\label{eq: lambda_k}
\end{equation}
\noindent Using Equation~(\ref{eq: lambda_k}), \ref{eq:bc_eq} reduces to
\begin{eqnarray}
C(\alpha,\beta,n)&=&\beta A {\displaystyle \sum_{i=1}^N}{\displaystyle \sum_{k=1}^n} \alpha^{k}{\lambda_{i}^{k}}Y_{i}\nonumber \\
&=& \beta A {\displaystyle \sum_{i=1}^N} \lambda_i  \frac{{(-1)}^{p_i}(1-\alpha^{n+1} \lambda_{i}^{n+1})}{{(-1)}^{p_i}(1-\alpha \lambda_{i})} Y_{i} \nonumber \\
\label{eq:cm}
\end{eqnarray}
\noindent where $p_{i}=0$ if $\alpha \left |\lambda_{i}\right | <1$, and $p_{i}=1$ if $\alpha \left |\lambda_{i}\right | >1$.
For the equations \ref{eq:bc_eq} and \ref{eq:cm}  to hold non-trivially,  $\alpha \neq 1/{\left |\lambda_{i}\right|},\ \forall i \in 1,2\cdots,N$.

We characterize the series \{$NC(\alpha,  \beta, n \to \infty)$\} for $\alpha \in [0,1]$ as follows:
 \begin{enumerate}
\item  $\alpha\ll\frac{1}{\left| \lambda_{1}  \right|}$: If $\alpha\ll\frac{1}{\left| \lambda_{1}  \right|}$ , $C(\alpha,\beta, n \to \infty)$ (and $NC(\alpha,\beta, n \to \infty)$ ) would be independent of $\alpha$, since
\begin{eqnarray}
C(\alpha,  \beta, n \to \infty) \approx \beta A \nonumber \\
NC(\alpha, \beta,  n \to \infty) \approx \frac{A}{\sum_{ij}(A)_{ij}}
\label{eq:b_alpha1}
\end{eqnarray}
\item   $\alpha<\frac{1}{\left |\lambda_{1}\right | }$: The sequence of matrices $\{ C(\alpha,\beta,n) \}$  \emph{converges} to $C(\alpha, \beta) $ as $n \to \infty $ if all the sequences $ \{{(C(\alpha ,\beta, n))}_{i j}\}$ for every  fixed $i$ and $j$ converge to ${(C(\alpha, \beta))}_{i j}$ ~\cite{Dienes:1932}.
If  $\alpha<\frac{1}{\left |\lambda_{1}\right | }$,  $ {C(\alpha, \beta, n)}$ \emph{converges} to $C(\alpha, \beta)$.
\begin{eqnarray}
C(\alpha, \beta, n \to \infty) & = & \beta A{\displaystyle \sum_{i=0}^N} \frac{\lambda_i}{1-\alpha \lambda_{i}} Y_{i} = {\beta A(I- \alpha A)^{-1}} \nonumber \\
        & = &  C(\alpha,\beta) \\
NC(\alpha, \beta, n \to \infty) & = & \frac{C(\alpha,\beta)}{\sum_{ij} {(C(\alpha, \beta))}_{ij}}
\label{eq:b_alpha}
\end{eqnarray}
\item  $\alpha>\frac{1}{\left |\lambda_{1}\right | } $ and $n \to \infty$, ${ \beta \alpha^{n}A^{n+1}}$  dominates in the Equation~(\ref{eq:cm}).
\begin{eqnarray}
C(\alpha, \beta, n\to \infty) \approx {\beta \alpha^{n}A^{n+1}} \nonumber \\
NC(\alpha, \beta, n\to \infty) \approx \frac{A^{n+1}}{\displaystyle \sum_{i,j} {A}^{n+1}_{ij}}
\label{eq:maxA}
\end{eqnarray}
\end{enumerate}

\begin{theorem}
\label{th1}
The induced ordering of nodes due to normalized $\alpha$-centrality is equal to the induced ordering of nodes due to $\alpha$-centrality for $\alpha <1/{\left |\lambda_{1}\right | }$.
\end{theorem}
\begin{proof}
Since centrality score due to $\alpha$-centrality is  $eC(\alpha, \beta,  n \to \infty)$ and  that due to normalized $\alpha$-centrality is $ eNC(\alpha,\beta, n \to \infty)$, from equations \ref{eq:b_alpha1} and \ref{eq:b_alpha}, the  induced ordering of nodes due to $\alpha$-centrality ($\alpha <1/{\left |\lambda_{1}\right | }$) would be equal to  induced ordering of nodes due to normalized $\alpha$-centrality  ($\alpha <1/{\left |\lambda_{1}\right | }$).
\end{proof}

\begin{theorem}
\label{th2}
The value of normalized $\alpha$-centrality matrix remains the same $\forall \alpha \in (1/{\left |\lambda_{1}\right | },1]$ ($NC(\alpha >1/{\left |\lambda_{1}\right | , \beta , n\to \infty} )=NC(\beta, n\to \infty)$).
\end{theorem}
\begin{proof}
As can be seen from equation \ref{eq:maxA}  when $\alpha>1/{\left |\lambda_{1}\right | } $ and $n \to \infty$, $NC(\alpha >1/{\left |\lambda_{1}\right | }, \beta, n \to \infty)$ reduces to  ${A^{n+1}}/{\displaystyle \sum_{i,j} {A}^{n+1}_{ij}} = NC(\beta, n\to \infty)$ and is independent of $\alpha$.
\end{proof}

 The remaining theorems hold under the condition that $|\lambda_{1}|$ is strictly greater than any other eigenvalue, which is true in most real life cases studied.
\begin{theorem}
\label{th3}
$\lim_{\alpha  \to 1/{|\lambda_{1}|}}NC_i( \alpha,\beta, n\to \infty) $ exists and $\lim_{\alpha  \to 1/{|\lambda_{1}|}}NC_i( \alpha,\beta, n\to \infty)  =NC_i(\beta, n\to \infty)=NC_i( \alpha >1/{|\lambda_{1}|} ,\beta, n\to \infty)  ={eAY_{1}}/{ \sum_{i,j} {(AY_{1})}_{ij}} $.
\end{theorem}
\begin{proof}
Under the assumption that $|\lambda_{1}|$ is strictly greater than any eigenvalue, as $\alpha \to 1/{|\lambda^{-}_{1}|}$, Equation~(\ref{eq:b_alpha})  reduces to
\begin{equation}
C(\alpha \to 1/{|\lambda^{-}_{1}| },  \beta, n \to \infty) \approx  \frac{\beta  \lambda_1}{1-\alpha \lambda_{1}}AY_{1}
\end{equation}
This is because all other eigenvectors shrink in importance as $\alpha \to 1/{|\lambda^{-}_{1}|}$ \cite{Bonacich:2001}. Therefore, as $\alpha \to 1/{|\lambda^{-}_{1}|}$ , we have
\begin{equation}
NC(\alpha  \to 1/{|\lambda^{-}_{1}|}, \beta, n \to \infty) \approx \frac{AY_{1}}{\displaystyle \sum_{i,j}{(AY_{1})}_{ij}}
\label{eq:max1}
\end{equation}
Under the assumption that $|\lambda_{1}|$ is strictly greater than any other eigenvalue, ${\beta \alpha^{n}\lambda_{1}^{n}AY_{1}}$  dominates in the Equation~(\ref{eq:cm}),  \ref{eq:maxA}.
\begin{eqnarray}
C(\alpha \to 1/{|\lambda^{+}_{1}|}, \beta, n\to \infty) \approx { \beta \alpha^{n}\lambda_{1}^{n}AY_{1}} \nonumber \\
NC(\alpha \to 1/{|\lambda^{+}_{1}|}, \beta, n\to \infty) \approx \frac{AY_{1}}{\displaystyle \sum_{i,j} {(AY_{1})}_{ij}}
\label{eq:max}
\end{eqnarray}

Hence from equation \ref{eq:max},  as $\alpha \to 1/{|\lambda^{+}_{1}|}$,  we have
\begin{equation}
NC(\alpha  \to 1/{|\lambda^{+}_{1}|}, \beta, n \to \infty) \approx \frac{AY_{1}}{\displaystyle \sum_{i,j}{(AY_{1})}_{ij}}
\label{eq:max2}
\end{equation}

Since,
\begin{eqnarray*}
\lim_{\alpha  \to 1/{|\lambda^{-}_{1}|}}NC(\alpha, \beta, n \to \infty)  & = & \lim_{\alpha  \to 1/{|\lambda^{+}_{1}|}}NC(\alpha, \beta,  n \to \infty) \\
& & =\frac{AY_{1}}{\sum_{i,j} {(AY_{1})}_{ij}},
\end{eqnarray*}
\noindent therefore, the limit
$\lim_{\alpha  \to 1/{|\lambda_{1}|}}NC(\alpha, \beta, n \to \infty)$ exists and
\begin{eqnarray}
\lim_{\alpha  \to 1/{|\lambda_{1}|}}NC(\alpha, \beta, n \to \infty) & = &\frac{AY_{1}}{ \sum_{i,j} {(AY_{1})}_{ij}} \nonumber \\
 & = & NC(\beta, n\to \infty) \nonumber
\label{eq:max3}
\end{eqnarray}
\noindent Since $NC_i( \alpha,\beta, n\to \infty)= eNC(\alpha, \beta,  n \to \infty)$, therefore,  $\lim_{\alpha  \to 1/{|\lambda_{1}|}}NC_i( \alpha,\beta, n\to \infty)  =NC_i(\beta, n\to \infty)=NC_i( \alpha >1/{|\lambda_{1}|} ,\beta, n\to \infty)  ={eAY_{1}}/{ \sum_{i,j} {(AY_{1})}_{ij}} $.

\end{proof}

\begin{theorem}
\label{th4}
For symmetric matrices, the  induced ordering of nodes due to eigenvector centrality $C_{E}$ is equivalent to the  induced ordering of nodes given by normalized centrality  $NC_i(\beta, n\to \infty)=\lim_{\alpha  \to 1/{|\lambda_{1}|}}NC_i( \alpha,\beta, n\to \infty)  =NC_i( \alpha >1/{|\lambda_{1}|} ,\beta, n\to \infty)  ={eAY_{1}}/{ \sum_{i,j}{(AY_{1})}_{ij}} $.
\end{theorem}
\begin{proof}
For symmetric matrices
\begin{equation}
A=X\Lambda X^{-1} = X \Lambda X^{T}
\label{eq:eigen1}
\end{equation}
Therefore equation \ref{eq:z} reduces to
\begin{equation}
Y_{i}=X {Z}_{i}X^{T} = X_{i}X_{i}^{T}
\end{equation}
\noindent where $X_{i}$ is the column of $X$ representing the eigenvector corresponding to $\lambda_{i}$.
Hence, in case of symmetric matrices:
\begin{eqnarray}
NC_i(\beta, n\to \infty)&=&NC_i( \alpha >1/{|\lambda_{1}|} ,\beta, n\to \infty) \nonumber\\
&=& \lim_{\alpha  \to 1/{|\lambda_{1}|}}NC_i( \alpha,\beta, n\to \infty)    \nonumber \\
&=& \frac{eAY_{1}}{ \sum_{i,j} {(AY_{1})}_{ij}}  \nonumber \\
&=& c_{1}eAX_{1}X_{1}^{T} = c_{2}X_{1}^{T}
\label{eq:max4}
\end{eqnarray}
where $c_{1}= \frac{1}{\sum_{i,j} {(AY_{1})}_{ij}}$ and $c_{2}= c_{1}eAX_{1}$.

Since $X_{1}^{T}$ corresponds to the eigenvector centrality vector $C_{E}$, hence for symmetric matrices, the induced ordering of nodes given by eigenvector centrality $C_{E}$ is equivalent to the  induced ordering of nodes given by normalized centrality $NC_i(\beta, n\to \infty)=\lim_{\alpha  \to 1/{|\lambda_{1}|}}NC_i( \alpha,\beta, n\to \infty)  =NC_i( \alpha >1/{|\lambda_{1}|} ,\beta, n\to \infty)  ={eAY_{1}}/{ \sum_{i,j} {(AY_{1})}_{ij}} $.
\end{proof}


\begin{thebibliography}{46}
\expandafter\ifx\csname natexlab\endcsname\relax\def\natexlab#1{#1}\fi
\expandafter\ifx\csname bibnamefont\endcsname\relax
  \def\bibnamefont#1{#1}\fi
\expandafter\ifx\csname bibfnamefont\endcsname\relax
  \def\bibfnamefont#1{#1}\fi
\expandafter\ifx\csname citenamefont\endcsname\relax
  \def\citenamefont#1{#1}\fi
\expandafter\ifx\csname url\endcsname\relax
  \def\url#1{\texttt{#1}}\fi
\expandafter\ifx\csname urlprefix\endcsname\relax\def\urlprefix{URL }\fi
\providecommand{\bibinfo}[2]{#2}
\providecommand{\eprint}[2][]{\url{#2}}

\bibitem[{\citenamefont{Bonacich}(2001)}]{Bonacich:2001}
\bibinfo{author}{\bibfnamefont{P.}~\bibnamefont{Bonacich}},
  \bibinfo{journal}{Social Networks} \textbf{\bibinfo{volume}{23}},
  \bibinfo{pages}{191} (\bibinfo{year}{2001}).

\bibitem[{\citenamefont{Newman and Girvan}(2004)}]{GirvanNewman04}
\bibinfo{author}{\bibfnamefont{M.~E.~J.} \bibnamefont{Newman}}
  \bibnamefont{and} \bibinfo{author}{\bibfnamefont{M.}~\bibnamefont{Girvan}},
  \bibinfo{journal}{Phys. Rev. E} \textbf{\bibinfo{volume}{69}},
  \bibinfo{pages}{026113} (\bibinfo{year}{2004}).

\bibitem[{\citenamefont{Freeman}(1979)}]{Freeman}
\bibinfo{author}{\bibfnamefont{L.~C.} \bibnamefont{Freeman}},
  \bibinfo{journal}{Social Networks} \textbf{\bibinfo{volume}{1}},
  \bibinfo{pages}{215} (\bibinfo{year}{1979}).

\bibitem[{\citenamefont{Stephenson and Zelen}(1989)}]{Stephenson89}
\bibinfo{author}{\bibfnamefont{K.}~\bibnamefont{Stephenson}} \bibnamefont{and}
  \bibinfo{author}{\bibfnamefont{M.}~\bibnamefont{Zelen}},
  \bibinfo{journal}{Social Networks} \textbf{\bibinfo{volume}{11}},
  \bibinfo{pages}{1} (\bibinfo{year}{1989}).

\bibitem[{\citenamefont{Page et~al.}(1998)\citenamefont{Page, Brin, Motwani,
  and Winograd}}]{PageRank}
\bibinfo{author}{\bibfnamefont{L.}~\bibnamefont{Page}},
  \bibinfo{author}{\bibfnamefont{S.}~\bibnamefont{Brin}},
  \bibinfo{author}{\bibfnamefont{R.}~\bibnamefont{Motwani}}, \bibnamefont{and}
  \bibinfo{author}{\bibfnamefont{T.}~\bibnamefont{Winograd}},
  \bibinfo{type}{Tech. Rep.}, \bibinfo{institution}{Stanford Digital Library
  Technologies Project} (\bibinfo{year}{1998}).

\bibitem[{\citenamefont{Noh and Rieger}(2002)}]{Noh02}
\bibinfo{author}{\bibfnamefont{J.~D.} \bibnamefont{Noh}} \bibnamefont{and}
  \bibinfo{author}{\bibfnamefont{H.}~\bibnamefont{Rieger}},
  \bibinfo{journal}{Phys. Rev. E} \textbf{\bibinfo{volume}{66}},
  \bibinfo{pages}{066127+} (\bibinfo{year}{2002}).

\bibitem[{\citenamefont{Newman}(2005)}]{Newman05}
\bibinfo{author}{\bibfnamefont{M.}~\bibnamefont{Newman}},
  \bibinfo{journal}{Social Networks} \textbf{\bibinfo{volume}{27}},
  \bibinfo{pages}{39} (\bibinfo{year}{2005}).

\bibitem[{\citenamefont{Katz}(1953)}]{Katz}
\bibinfo{author}{\bibfnamefont{L.}~\bibnamefont{Katz}},
  \bibinfo{journal}{Psychometrika} \textbf{\bibinfo{volume}{18}},
  \bibinfo{pages}{39} (\bibinfo{year}{1953}).

\bibitem[{\citenamefont{Bonacich}(1987)}]{Bonacich87}
\bibinfo{author}{\bibfnamefont{P.}~\bibnamefont{Bonacich}},
  \bibinfo{journal}{Am. J. Sociology} \textbf{\bibinfo{volume}{92}},
  \bibinfo{pages}{1170} (\bibinfo{year}{1987}).

\bibitem[{\citenamefont{Granovetter}(1973)}]{Granovetter}
\bibinfo{author}{\bibfnamefont{M.}~\bibnamefont{Granovetter}},
  \bibinfo{journal}{Am. J. Sociology}  (\bibinfo{year}{1973}).

\bibitem[{\citenamefont{Simmel}(1950)}]{Simmel}
\bibinfo{author}{\bibfnamefont{G.}~\bibnamefont{Simmel}},
  \emph{\bibinfo{title}{The Sociology of Georg Simmel}}
  (\bibinfo{publisher}{Free Press}, \bibinfo{year}{1950}), chap.
  \bibinfo{chapter}{Individual and Society}.

\bibitem[{\citenamefont{Csermely}(2008)}]{Csermely08}
\bibinfo{author}{\bibfnamefont{P.}~\bibnamefont{Csermely}},
  \bibinfo{journal}{Trends in Biochemical Sciences}
  \textbf{\bibinfo{volume}{33}}, \bibinfo{pages}{569} (\bibinfo{year}{2008}).

\bibitem[{\citenamefont{Lusseau and Newman}(2004)}]{Lusseau04}
\bibinfo{author}{\bibfnamefont{D.}~\bibnamefont{Lusseau}} \bibnamefont{and}
  \bibinfo{author}{\bibfnamefont{M.~E.~J.} \bibnamefont{Newman}},
  \bibinfo{journal}{Proc. Royal Society of London. Series B: Biological
  Sciences} \textbf{\bibinfo{volume}{271}}, \bibinfo{pages}{S477}
  (\bibinfo{year}{2004}).

\bibitem[{\citenamefont{Fortunato}(2010)}]{Fortunato10}
\bibinfo{author}{\bibfnamefont{S.}~\bibnamefont{Fortunato}},
  \bibinfo{journal}{Phys. Reports} \textbf{\bibinfo{volume}{486}},
  \bibinfo{pages}{75} (\bibinfo{year}{2010}).

\bibitem[{\citenamefont{Newman}(2004{\natexlab{a}})}]{Newman104}
\bibinfo{author}{\bibfnamefont{M.~E.~J.} \bibnamefont{Newman}},
  \bibinfo{journal}{Phys. Rev. E} \textbf{\bibinfo{volume}{69}},
  \bibinfo{pages}{066133} (\bibinfo{year}{2004}{\natexlab{a}}).

\bibitem[{\citenamefont{Newman}(2006)}]{Newman206}
\bibinfo{author}{\bibfnamefont{M.~E.~J.} \bibnamefont{Newman}},
  \bibinfo{journal}{Phys. Rev. E} \textbf{\bibinfo{volume}{74}},
  \bibinfo{pages}{036104} (\bibinfo{year}{2006}).

\bibitem[{\citenamefont{Milo et~al.}(2002)\citenamefont{Milo, Shen-Orr,
  Itzkovitz, Kashtan, Chklovskii, and Alon}}]{Milo02}
\bibinfo{author}{\bibfnamefont{R.}~\bibnamefont{Milo}},
  \bibinfo{author}{\bibfnamefont{S.}~\bibnamefont{Shen-Orr}},
  \bibinfo{author}{\bibfnamefont{S.}~\bibnamefont{Itzkovitz}},
  \bibinfo{author}{\bibfnamefont{N.}~\bibnamefont{Kashtan}},
  \bibinfo{author}{\bibfnamefont{D.}~\bibnamefont{Chklovskii}},
  \bibnamefont{and} \bibinfo{author}{\bibfnamefont{U.}~\bibnamefont{Alon}},
  \bibinfo{journal}{Science} \textbf{\bibinfo{volume}{298}},
  \bibinfo{pages}{824} (\bibinfo{year}{2002}).

\bibitem[{\citenamefont{Milo et~al.}(2004)\citenamefont{Milo, Itzkovitz,
  Kashtan, Levitt, Shen-Orr, Ayzenshtat, Sheffer, and Alon}}]{Milo04}
\bibinfo{author}{\bibfnamefont{R.}~\bibnamefont{Milo}},
  \bibinfo{author}{\bibfnamefont{S.}~\bibnamefont{Itzkovitz}},
  \bibinfo{author}{\bibfnamefont{N.}~\bibnamefont{Kashtan}},
  \bibinfo{author}{\bibfnamefont{R.}~\bibnamefont{Levitt}},
  \bibinfo{author}{\bibfnamefont{S.}~\bibnamefont{Shen-Orr}},
  \bibinfo{author}{\bibfnamefont{I.}~\bibnamefont{Ayzenshtat}},
  \bibinfo{author}{\bibfnamefont{M.}~\bibnamefont{Sheffer}}, \bibnamefont{and}
  \bibinfo{author}{\bibfnamefont{U.}~\bibnamefont{Alon}},
  \bibinfo{journal}{Science} \textbf{\bibinfo{volume}{303}},
  \bibinfo{pages}{1538} (\bibinfo{year}{2004}).

\bibitem[{\citenamefont{Guimera and Amaral}(2005)}]{Guimera05}
\bibinfo{author}{\bibfnamefont{R.}~\bibnamefont{Guimera}} \bibnamefont{and}
  \bibinfo{author}{\bibfnamefont{L.~A.~N.} \bibnamefont{Amaral}},
  \bibinfo{journal}{Nature} \textbf{\bibinfo{volume}{433}},
  \bibinfo{pages}{895} (\bibinfo{year}{2005}).

\bibitem[{\citenamefont{Guimera et~al.}(2007)\citenamefont{Guimera,
  Sales-Pardo, and Amaral}}]{Guimera07}
\bibinfo{author}{\bibfnamefont{R.}~\bibnamefont{Guimera}},
  \bibinfo{author}{\bibfnamefont{M.}~\bibnamefont{Sales-Pardo}},
  \bibnamefont{and} \bibinfo{author}{\bibfnamefont{L.~A.~N.}
  \bibnamefont{Amaral}}, \bibinfo{journal}{Nat Phys}
  \textbf{\bibinfo{volume}{3}}, \bibinfo{pages}{63} (\bibinfo{year}{2007}).

\bibitem[{\citenamefont{Ferrar}(1951)}]{Ferrar}
\bibinfo{author}{\bibfnamefont{W.~L.} \bibnamefont{Ferrar}},
  \emph{\bibinfo{title}{Finite Matrices}} (\bibinfo{publisher}{Oxford Univ.
  Press}, \bibinfo{year}{1951}).

\bibitem[{\citenamefont{Haveliwala}(1999)}]{Haveliwala:1999}
\bibinfo{author}{\bibfnamefont{T.~H.} \bibnamefont{Haveliwala}},
  \bibinfo{type}{Tech. Rep.}, \bibinfo{institution}{Stanford University}
  (\bibinfo{year}{1999}).

\bibitem[{\citenamefont{Kamvar et~al.}(2003)\citenamefont{Kamvar, Haveliwala,
  Manning, and Golub}}]{Kamvar:2003}
\bibinfo{author}{\bibfnamefont{S.}~\bibnamefont{Kamvar}},
  \bibinfo{author}{\bibfnamefont{T.}~\bibnamefont{Haveliwala}},
  \bibinfo{author}{\bibfnamefont{C.}~\bibnamefont{Manning}}, \bibnamefont{and}
  \bibinfo{author}{\bibfnamefont{G.}~\bibnamefont{Golub}}, \bibinfo{type}{Tech.
  Rep.}, \bibinfo{institution}{Stanford University} (\bibinfo{year}{2003}).

\bibitem[{\citenamefont{Dean and Ghemawat}(2008)}]{Dean:2008}
\bibinfo{author}{\bibfnamefont{J.}~\bibnamefont{Dean}} \bibnamefont{and}
  \bibinfo{author}{\bibfnamefont{S.}~\bibnamefont{Ghemawat}},
  \bibinfo{journal}{Commun. ACM} \textbf{\bibinfo{volume}{51}},
  \bibinfo{pages}{107} (\bibinfo{year}{2008}), ISSN \bibinfo{issn}{0001-0782}.

\bibitem[{\citenamefont{Wasserman and Faust}(1994)}]{Wasserman:1994}
\bibinfo{author}{\bibfnamefont{S.}~\bibnamefont{Wasserman}} \bibnamefont{and}
  \bibinfo{author}{\bibfnamefont{K.}~\bibnamefont{Faust}},
  \emph{\bibinfo{title}{Social Network Analysis: Methods and Applications}}
  (\bibinfo{publisher}{Cambridge Univ.Press}, \bibinfo{year}{1994}).

\bibitem[{\citenamefont{Burt}(1992)}]{Burt}
\bibinfo{author}{\bibfnamefont{R.~S.} \bibnamefont{Burt}},
  \emph{\bibinfo{title}{Structural Holes: The Structure of Competition}}
  (\bibinfo{publisher}{Harvard University Press}, \bibinfo{address}{Cambridge,
  MA}, \bibinfo{year}{1992}).

\bibitem[{\citenamefont{Newman}(2004{\natexlab{b}})}]{Newman204}
\bibinfo{author}{\bibfnamefont{M.~E.~J.} \bibnamefont{Newman}},
  \bibinfo{journal}{The European Physical Journal B}
  \textbf{\bibinfo{volume}{38}}, \bibinfo{pages}{321}
  (\bibinfo{year}{2004}{\natexlab{b}}).

\bibitem[{\citenamefont{Ghosh and Lerman}(2008)}]{Ghosh08}
\bibinfo{author}{\bibfnamefont{R.}~\bibnamefont{Ghosh}} \bibnamefont{and}
  \bibinfo{author}{\bibfnamefont{K.}~\bibnamefont{Lerman}}, in
  \emph{\bibinfo{booktitle}{Proc. 2nd KDD Workshop on Social Network Analysis
  (SNAKDD'08)}} (\bibinfo{year}{2008}).

\bibitem[{\citenamefont{Brandes et~al.}(2008)\citenamefont{Brandes, Delling,
  Gaertler, Gorke, Hoefer, Nikoloski, and Wagner}}]{Brandes}
\bibinfo{author}{\bibfnamefont{U.}~\bibnamefont{Brandes}},
  \bibinfo{author}{\bibfnamefont{D.}~\bibnamefont{Delling}},
  \bibinfo{author}{\bibfnamefont{M.}~\bibnamefont{Gaertler}},
  \bibinfo{author}{\bibfnamefont{R.}~\bibnamefont{Gorke}},
  \bibinfo{author}{\bibfnamefont{M.}~\bibnamefont{Hoefer}},
  \bibinfo{author}{\bibfnamefont{Z.}~\bibnamefont{Nikoloski}},
  \bibnamefont{and} \bibinfo{author}{\bibfnamefont{D.}~\bibnamefont{Wagner}},
  \bibinfo{journal}{IEEE Trans. on Knowl. and Data Eng.}
  \textbf{\bibinfo{volume}{20}}, \bibinfo{pages}{172} (\bibinfo{year}{2008}),
  ISSN \bibinfo{issn}{1041-4347}.

\bibitem[{\citenamefont{Leicht and Newman}(2008)}]{Leicht}
\bibinfo{author}{\bibfnamefont{E.~A.} \bibnamefont{Leicht}} \bibnamefont{and}
  \bibinfo{author}{\bibfnamefont{M.~E.~J.} \bibnamefont{Newman}},
  \bibinfo{journal}{Phys. Rev. Letters} \textbf{\bibinfo{volume}{100}},
  \bibinfo{pages}{118703} (\bibinfo{year}{2008}).

\bibitem[{\citenamefont{Tong et~al.}(2006)\citenamefont{Tong, Faloutsos, and
  Pan}}]{Tong06}
\bibinfo{author}{\bibfnamefont{H.}~\bibnamefont{Tong}},
  \bibinfo{author}{\bibfnamefont{C.}~\bibnamefont{Faloutsos}},
  \bibnamefont{and} \bibinfo{author}{\bibfnamefont{J.}~\bibnamefont{Pan}},
  \bibinfo{journal}{Data Mining, 2006. ICDM '06. Sixth International Conference
  on} pp. \bibinfo{pages}{613--622} (\bibinfo{year}{2006}), ISSN
  \bibinfo{issn}{1550-4786}.

\bibitem[{\citenamefont{Tong et~al.}(2008)\citenamefont{Tong, Papadimitriou,
  Yu, and Faloutsos}}]{Tong08}
\bibinfo{author}{\bibfnamefont{H.}~\bibnamefont{Tong}},
  \bibinfo{author}{\bibfnamefont{S.}~\bibnamefont{Papadimitriou}},
  \bibinfo{author}{\bibfnamefont{P.~S.} \bibnamefont{Yu}}, \bibnamefont{and}
  \bibinfo{author}{\bibfnamefont{C.}~\bibnamefont{Faloutsos}}, in
  \emph{\bibinfo{booktitle}{Proc. SIAM Conference on Data Mining}}
  (\bibinfo{year}{2008}), pp. \bibinfo{pages}{704--715}.

\bibitem[{\citenamefont{Zhou}(2003)}]{Zhou03}
\bibinfo{author}{\bibfnamefont{H.}~\bibnamefont{Zhou}}, \bibinfo{journal}{Phys.
  Rev. E} \textbf{\bibinfo{volume}{67}}, \bibinfo{pages}{041908}
  (\bibinfo{year}{2003}).

\bibitem[{\citenamefont{Zachary}(1977)}]{Zachary}
\bibinfo{author}{\bibfnamefont{W.~W.} \bibnamefont{Zachary}},
  \bibinfo{journal}{J. Anthropological Research} \textbf{\bibinfo{volume}{33}},
  \bibinfo{pages}{452} (\bibinfo{year}{1977}).

\bibitem[{\citenamefont{Duch and Arenas}(2005)}]{Duch05}
\bibinfo{author}{\bibfnamefont{J.}~\bibnamefont{Duch}} \bibnamefont{and}
  \bibinfo{author}{\bibfnamefont{A.}~\bibnamefont{Arenas}},
  \bibinfo{journal}{Phys. Rev. E} \textbf{\bibinfo{volume}{72}},
  \bibinfo{pages}{027104+} (\bibinfo{year}{2005}).

\bibitem[{\citenamefont{Padgett and Ansell}(1993)}]{Padgett}
\bibinfo{author}{\bibfnamefont{J.~F.} \bibnamefont{Padgett}} \bibnamefont{and}
  \bibinfo{author}{\bibfnamefont{C.~K.} \bibnamefont{Ansell}},
  \bibinfo{journal}{Am. J. Sociology} \textbf{\bibinfo{volume}{98}},
  \bibinfo{pages}{1259} (\bibinfo{year}{1993}).

\bibitem[{\citenamefont{Girvan and Newman}(2002)}]{GirvanNewman02}
\bibinfo{author}{\bibfnamefont{M.}~\bibnamefont{Girvan}} \bibnamefont{and}
  \bibinfo{author}{\bibfnamefont{M.~E.~J.} \bibnamefont{Newman}},
  \bibinfo{journal}{Proc. Natl. Acad. Sci. USA.} \textbf{\bibinfo{volume}{99}},
  \bibinfo{pages}{7821} (\bibinfo{year}{2002}).

\bibitem[{\citenamefont{L.Wallace}(1983)}]{Wallace83}
\bibinfo{author}{\bibfnamefont{D.}~\bibnamefont{L.Wallace}},
  \bibinfo{journal}{Journal of the American Statistical Association}
  \textbf{\bibinfo{volume}{383}}, \bibinfo{pages}{569} (\bibinfo{year}{1983}).

\bibitem[{\citenamefont{Hubbel}(1965)}]{Hubbell:1965}
\bibinfo{author}{\bibfnamefont{C.}~\bibnamefont{Hubbel}},
  \bibinfo{journal}{Sociometry} \textbf{\bibinfo{volume}{28}},
  \bibinfo{pages}{377} (\bibinfo{year}{1965}).

\bibitem[{\citenamefont{Bonacich}(1972)}]{Bonacich:1972}
\bibinfo{author}{\bibfnamefont{P.}~\bibnamefont{Bonacich}},
  \bibinfo{journal}{J. of Mathematical Sociology} \textbf{\bibinfo{volume}{2}},
  \bibinfo{pages}{113} (\bibinfo{year}{1972}).

\bibitem[{\citenamefont{Liben-Nowell and Kleinberg}(2007)}]{Liben07}
\bibinfo{author}{\bibfnamefont{D.}~\bibnamefont{Liben-Nowell}}
  \bibnamefont{and}
  \bibinfo{author}{\bibfnamefont{J.}~\bibnamefont{Kleinberg}},
  \bibinfo{journal}{J. Am. Soc. Inf. Sci. Technol.}
  \textbf{\bibinfo{volume}{58}}, \bibinfo{pages}{1019} (\bibinfo{year}{2007}),
  ISSN \bibinfo{issn}{1532-2882}.

\bibitem[{\citenamefont{Ghosh and Lerman}(2010)}]{Ghosh10snakdd}
\bibinfo{author}{\bibfnamefont{R.}~\bibnamefont{Ghosh}} \bibnamefont{and}
  \bibinfo{author}{\bibfnamefont{K.}~\bibnamefont{Lerman}}, in
  \emph{\bibinfo{booktitle}{Proc. KDD workshop on Social Network Analysis
  (SNA-KDD)}} (\bibinfo{year}{2010}).

\bibitem[{\citenamefont{Arenas et~al.}(2008)\citenamefont{Arenas, Fernandez,
  Fortunato, and Gomez}}]{Arenas}
\bibinfo{author}{\bibfnamefont{A.}~\bibnamefont{Arenas}},
  \bibinfo{author}{\bibfnamefont{A.}~\bibnamefont{Fernandez}},
  \bibinfo{author}{\bibfnamefont{S.}~\bibnamefont{Fortunato}},
  \bibnamefont{and} \bibinfo{author}{\bibfnamefont{S.}~\bibnamefont{Gomez}},
  \bibinfo{journal}{Mathematical Systems Theory} \textbf{\bibinfo{volume}{41}},
  \bibinfo{pages}{224001} (\bibinfo{year}{2008}).

\bibitem[{\citenamefont{Ghosh and Lerman}(2009)}]{Ghosh09socialcom}
\bibinfo{author}{\bibfnamefont{R.}~\bibnamefont{Ghosh}} \bibnamefont{and}
  \bibinfo{author}{\bibfnamefont{K.}~\bibnamefont{Lerman}}, in
  \emph{\bibinfo{booktitle}{Proc. 1st IEEE SIGCOM Social Computing Conference
  (SocialCom09)}} (\bibinfo{year}{2009}).

\bibitem[{\citenamefont{Gebali}(2008)}]{Gebali:2008}
\bibinfo{author}{\bibfnamefont{F.}~\bibnamefont{Gebali}},
  \bibinfo{journal}{Analysis of Computer and Communication Networks} p.
  \bibinfo{pages}{65:122} (\bibinfo{year}{2008}).

\bibitem[{\citenamefont{Dienes}(1932)}]{Dienes:1932}
\bibinfo{author}{\bibfnamefont{P.}~\bibnamefont{Dienes}},
  \bibinfo{journal}{Quart. J. of Math. (Oxford)} \textbf{\bibinfo{volume}{3}},
  \bibinfo{pages}{253} (\bibinfo{year}{1932}).

\end{thebibliography}
\end{document}